\documentclass[10pt]{article}
\pdfoutput=1
\usepackage{amssymb} % that is for real numbers etc.
\usepackage{amsmath}
\usepackage{amsthm}   % enables to number all with 1 single counter
\usepackage{epsfig}
\usepackage{multirow}
\usepackage{natbib}
\usepackage{color}
 \usepackage{a4wide}

\usepackage{booktabs}
\usepackage{placeins}
\usepackage{colortbl}
\usepackage{caption}
\usepackage{subcaption}
\usepackage{makecell}
\usepackage[dvipsnames]{xcolor}

\usepackage[british]{babel}
\bibpunct{(}{)}{;}{a}{}{,}

\newtheorem{Prop}{Proposition}[section]
\newtheorem{Test}[Prop]{Test}
\newtheorem{Lem}[Prop]{Lemma}
\newtheorem{Th}[Prop]{Theorem}
\newtheorem{Rm}[Prop]{Remark}
\newtheorem{Def}[Prop]{Definition}
\newtheorem{Cor}[Prop]{Corollary}

\newtheorem{Proc}[Prop]{Procedure}

\newfont{\smcal}{cmu10 scaled 1200}
\newfont{\handw}{cmmi10 scaled 1200}
\newfont{\handws}{cmmi10 scaled 800}

\newcommand{\rank}{\mbox{\rm rank}}

\newcommand{\RR}{{\mathbb R}}

\newcommand{\diag}{\mbox{\rm diag}}

\newcommand{\tr}{\mbox{\rm trace}}

\newcommand{\sinc}{\mbox{\rm sinc}}
\newcommand{\argmin}{\mathop{\mbox{\rm argmin}}}
\newcommand{\argmax}{\mathop{\mbox{\rm argmax}}}

\newcommand{\arginf}{\mathop{\mbox{\rm arginf}}}

\newcommand{\pr}{\mathfrak{pr}}

\newcommand{\iid}{{\stackrel{i.i.d.}{\sim}}}
\newcommand{\X}{{\cal X}}%
\newcommand{\R}{{\mathbb R}}%
\newcommand{\SSS}{{\mathbb S}}%
\newcommand{\E}{{\mathbb E}}%

\newcommand{\N}{\mathbb{N}}

\newcommand{\wL}{\widetilde{L}}
\newcommand{\Exp}{\mbox{\rm Exp}}
\newcommand{\Log}{\mbox{\rm Log}}
\newcommand{\Inull}{\mathcal{I}_0\big(G\big)}
\newcommand{\Diff}{\rm{Diff}^+[0,1]}

\newcommand{\so}{\mathfrak{so}}
\newcommand{\fg}{\mathfrak{g}}
\newcommand{\G}{G}

\newcommand{\Shape}{\mathcal{S}}

\newcommand{\cC}{\mathcal{C}}
\newcommand{\cG}{\mathcal{G}}

\newcommand{\cK}{\mathcal{K}}

\newcommand{\Prb}{\mathbb{P}}
\newcommand{\Norm}{\mathcal{N}}

\title{Functional Inference on Rotational Curves and Identification of Human Gait at the Knee Joint}
\author{Fabian J.E. Telschow\footnote{Felix Bernstein Institute for Mathematical Statistics in the Biosciences, Georgia Augusta University of G\"ottingen, Goldschmidtstrasse 7, 37077 G\"ottingen, Germany}, 
 Stephan F. Huckemann$^*$ and Michael R. Pierrynowski\footnote{School
 of Rehabilitation Science, McMaster University, Canada
}}

\begin{document}
\date{}
\maketitle
% \tableofcontents

% % % % % % % % % % % % % % % % % % % % % % % % % % % Abstract % % % % % % % % % % % % % % % % % % % % % % % % % % %
% I've made a few small changes to the abstract. I also added sentence 2. 
% \pagebreak
\begin{abstract}

We extend Gaussian perturbation models in classical functional data analysis to the three-dimensional rotational group where a zero-mean Gaussian process in the Lie algebra under the Lie exponential spreads multiplicatively around a central curve. As an estimator, we introduce point-wise extrinsic mean curves which feature strong perturbation consistency, and which are asymptotically a.s. unique and differentiable, if the model is so. Further, we consider the group action of time warping and that of spatial isometries that are connected to the identity. The latter can be asymptotically consistently estimated if lifted to the unit quaternions. Introducing a generic loss for Lie groups, the former can be estimated, and based on curve length, due to asymptotic differentiability, we propose two-sample permutation tests involving various combinations of the group actions. This methodology allows inference on gait patterns due to the rotational motion of the lower leg with respect to the upper leg. This was previously not possible because, among others, the usual analysis of separate Euler angles is not independent of marker placement, even if performed by trained specialists.
\end{abstract}

% Keywords
% Functional data analysis, Lie groups, modulo group actions, quaternions, Gaussian perturbation models

% % % % % % % % % % % % % % % % % % % % % % % % % % % Introduction % % % % % % % % % % % % % % % % % % % % % % % % % % %
\section{Introduction}

To date a rich set of powerful descriptive and inferential tools are available for the statistical goals of \emph{functional data analysis} (FDA) as coined by \cite{Ramsay1982}.  
For a recent overview, see \cite{Wang2015}. 
Beyond classical Euclidean FDA, currently, functional data taking values in manifolds is gaining increased interest. For instance the work of \cite{SrivastavaKlassen2011} has spurred several subsequent publications, e.g. \cite{Celledoni2015,Bauer2015,Bauer2016, Amor2016}. At this point, however, inferential tools, for manifold valued curves, say, are not yet available, although, as an important step, a general time warping method has been developed by \cite{SuKurtek2014}.

In this paper we propose a nonparametric two-sample test for functional data that takes values in the rotational group $SO(3)$. Like a classical two sample test it relies on the concept of a mean, in our case, a mean curve. With this endeavor, there are three challenges. 

First, unlike Euclidean spaces, in any of the non-Euclidean geometries of $SO(3)$ there is no a priori guarantee that the mean is unique and that its estimators feature favorable statistical properties. Second, it is not clear which from  the ``zoo'' of generalizations (e.g. \cite{Huckemann12}) of the classical Euclidean mean to a curved space should be taken. Third, temporal and spatial registration require specific non-Euclidean loss functions, that should at least be  invariant under reparametrization (temporal) and actions of isometries (spatial).

To these ends, naturally generalizing the common functional model in $\R^D$
\begin{equation}\label{eq:RDfunctionalModel}
 y(t) = \mu(t) + Z_t\,,
\end{equation}
where $\mu$ is deterministic and $Z_t$ is a zero-mean multivariate Gaussian process,  
we develop the framework of \emph{Gaussian perturbation} (GP) models 
\begin{equation}\label{eq:GPModel}
 y(t) = \mu(t)\Exp\big(Z_t\big)\,,
\end{equation}
in which a zero-mean Gaussian process $Z_t$ in the Lie algebra under the Lie exponential  $\Exp$ spreads multiplicatively around a central curve $\mu$. This is indeed a generalization because $\R^D$ can be viewed as an additive Lie group agreeing with its Lie algebra, such that the Lie exponential is the identity and the ``multiplication'' in (\ref{eq:GPModel}) is then just the summation in (\ref{eq:RDfunctionalModel}). In fact, we will see that our GP models are canonical in that way that the class of models due to multiplication from the right as in (\ref{eq:GPModel}) is equal to class of models due to multiplication from the left and asymptotically, as the noise level vanishes, equal to models due to multiplication from both sides. In non-functional settings analog models were studied and used in \citet{Downs1972}, \citet{Rancourt2000} and \citet{Fletcher2011}.

Defining the \emph{pointwise extrinsic Fr\'echet mean} (PEM) curve we show strong asymptotic consistency, namely that PEM curves converge a.s. to the central curve. This is rather unusual as \emph{perturbation consistency} is often not given, even if the noise is isotropic, e.g. \cite{Kent1997,Huckemann2011b}. Here it holds under very general conditions, in particular for our GP models which do not require isotropicity. Moreover, in view of our tests based on curve lengths, thus requiring existence of derivatives or at least absolute continuity, we show that for a.s.  $\cC^1$ (continuously differentiable) GP models, asymptotically sample PEM curves are again $\cC^1$.  

Further we model spatial registration by the action of the identity component of the isometry group on $SO(3)$. In a Euclidean geometry, these correspond to orientation preserving Euclidean motions. While for this action, a new loss induced from lifting to the unit quaternions, allows for asymptotically strongly consistent registration, % extrinsic While for this action, for computational feasibility we use an extrinsic loss function, 
for temporal alignment we define another new loss function which is canonical on any Lie group. Thus we can avoid the generic method for temporal alignment on manifolds proposed by \cite{SuKurtek2014} which suffers from non-canonical choices such as parallel transport to a reference tangent space. 
%As the loss measures curve length it requires derivatives or at least absolute continuity. We show that PEM curves are asymptotically $\cC^1$ a.s. if stemming from a $\cC^1$ GP model.
These facts combined with the equivariance of GP models under both group actions are used to formulate nonparametric two-sample permutation tests with and without spatial and/or temporal registration based on lengths of sample PEM curves.

Notably, some of our ideas generalize to $SO(n)$ (for instance, for odd dimension $n$, the group is still simple and hence the identity component of the group of spatial isometries is again given by $SO(n)\times SO(n)$, cf. Lemma \ref{lem:isometry-group}) and GP models (\ref{eq:GPModel}) in conjunction with the loss functions and their temporal invariance and inverse alignment properties can, of course, be defined on any Lie group. For a concise presentation, we have restricted ourselves to $SO(3)$
which allows fast spatial registration using the double cover of rotations by unit quaternions.
%however, as, for instance, fast spatial registration using the double cover of rotations by unit quaternions is only possible for $n=3$. 

This research has been motivated by the long standing problem of reproducibility of gait curves in the biomechanical analysis of motion (see \cite{duhamel2004statistical}). Here, studying trajectories of the relative rotations between tibia (the lower leg's main bone) with respect to the femur (the upper leg's bone) originating from typical walking motion is of interest, for instance in order to assess various degenerative processes (e.g., early onset of osteoarthritis, see \cite{pierrynowski2010patients}), as well as to evaluate therapeutic interventions (see \cite {ounpuu2015long}).

% \begin{figure}[!h] \centering
%    \subcaptionbox{\it Typical biomechanical analysis.}[0.48\textwidth]{\includegraphics[width=0.48\textwidth]{pics/Vol4raw.pdf}}
%    \subcaptionbox{\it After temporal and spatial registration.}[0.48\textwidth]{\includegraphics[width=0.48\textwidth]{pics/Vol4SSASTW.pdf}}
% %    \includegraphics[width=0.48\textwidth]{pics/LWV44S34_SSA_STW.pdf}
%    \caption{\it Depicting two samples of three Euler angle curves (top: flexion-extension, middle: abduction-adduction and bottom: internal-external rotation) determining motion at the knee joint for the same subject (volunteer 4) before (red) and after specialist marker replacement (blue). Bold face curves represent sample PEM curves. Left (a): raw data as processed using typical biomechanical analysis. Note the vertical shift of the bottom internal-external rotation angles. Right (b): After temporal and spatial registration, the vertical shift has been removed along with other smaller distortions.}
%    \label{fig:marker-replacement-effect}
% \end{figure}

\begin{figure}[!h] \centering
   \subcaptionbox{\it Typical biomechanical analysis.}[0.48\textwidth]{\includegraphics[width=0.48\textwidth]{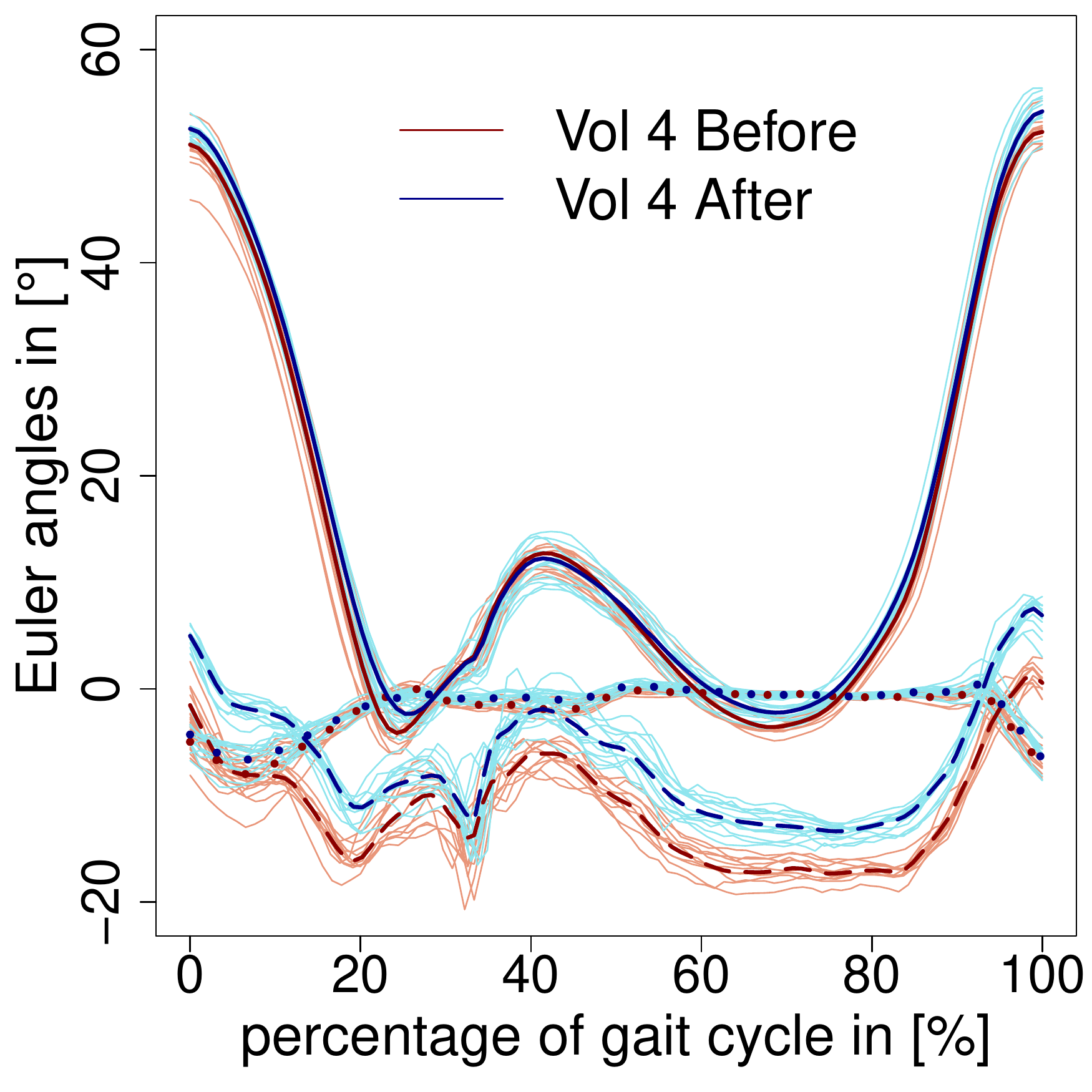}}
   \subcaptionbox{\it After temporal and spatial registration.}[0.48\textwidth]{\includegraphics[width=0.48\textwidth]{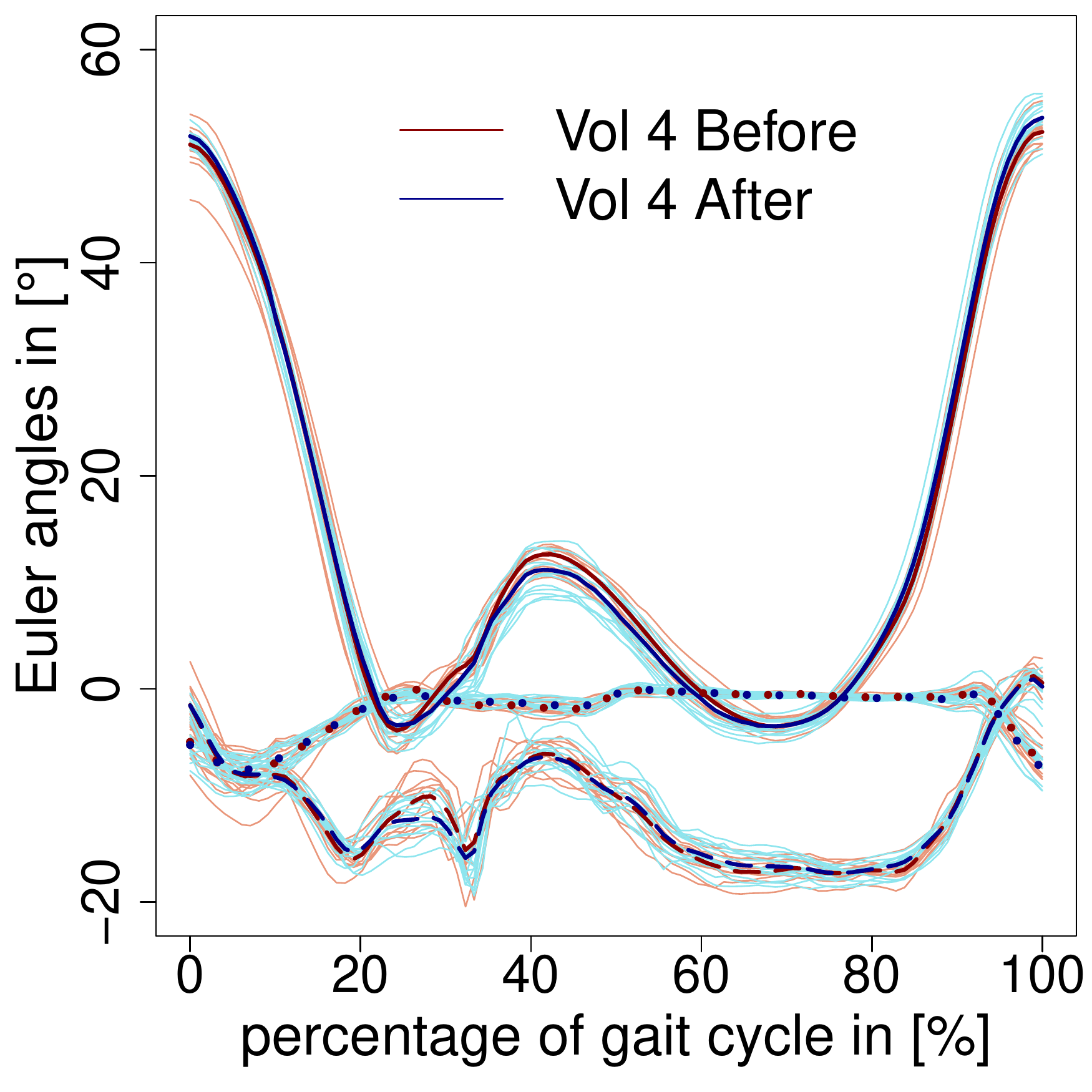}}
   \caption{\it Depicting two samples of three Euler angle curves (top: flexion-extension, middle: abduction-adduction and bottom: internal-external rotation) determining motion at the knee joint for the same subject (volunteer 4) before (red) and after specialist marker replacement (blue). Bold face curves represent sample PEM curves. Left (a): raw data as processed using typical biomechanical analysis. Note the vertical shift of the bottom internal-external rotation angles. Right (b): After temporal and spatial registration, the vertical shift has been removed along with other smaller distortions.}
   \label{fig:marker-replacement-effect}
\end{figure}

Unless invasive methods such as fluoroscopic (X-ray) imaging are used, in common practice the underlying bone motion must be estimated. Markers are placed on specific skin locations of subjects, and the 3D motion of these markers are recorded while subjects perform certain walking tasks. It is well known, however, that such methods introduce \emph{soft tissue artifacts} (e.g., \cite{Leardini2005}) because they also include skin, muscle and fat tissue motion which are superimposed over the unknown bone's motion.

We identify here another effect that influences the measurement of the underlying bone motion. Specialists place markers on anatomically defined skin locations and a calibration of relative coordinate frames is performed. However, the ability of specialists to place these markers repeatedly is questionable (see \cite{noehren2010improving} and \cite{groen2012sensitivity}). In an experiment performed we see a notable deviation between the three Euler angle curves of relative rotations, as defined in \cite{grood1983joint}, before and after specialist marker replacement, see Figure \ref{fig:marker-replacement-effect}a. Since the bottom internal-external rotation angles are considerably vertically shifted, a test based on all three Euler angles would reject the identity of the same subject with itself after trained specialist replacement. It seems that this effect has hampered the development of testing protocols that can reliably identify subjects when nothing changed but specialist \emph{marker placements} (MP). However, such identifications are a sine qua non for detecting changes in gait patterns due to degenerative effects or evaluate the effects of therapy. 
We will see that our permutation tests including spatial-temporal registration robustly remove this MP effect, see Figure \ref{fig:marker-replacement-effect}b, and allow for the development of inferential protocols in the future. These could then be applied to biomechanical analysis at other body joints as well.
% For an experimental data set collected at the School of Rehabilitation Science (Mc. Master Univ., Canada) {\bf we only take walk left} it turns out that our two-sample test with spatial registration precisely removes this MP effect, see Figure \ref{fig:marker-replacement-effect} (b). On this basis, identification of individuals is possible, robust under MP, and  

Our paper is  organized as follows. In Section \ref{scn:FunctionalSO3data}, after some background information on geometry, we introduce the data space with the two group actions and define $\cC^1$ GP models with a discussion on their canonicity. We then introduce Fr\'echet mean estimators and derive their benign properties in GP models. %, such as perturbation consistency and asymptotic a.s. uniqueness and differentiability. 
Section \ref{scn:EsimtationGaitSimilarities} deals with the estimation of spatial and temporal registration. %, where for the former strong asymptotic consistency is shown. 
Having provided for all of the ingredients, we propose our two-sample tests in Section \ref{scn:Testing}, provide for simulations in Section \ref{scn:Simulations} and apply in Section \ref{scn:Application} our new methodology to a biomechanical experiment conducted in the School of Rehabilitation Science at McMaster University (Canada). 

Unless otherwise stated, all of the proofs have been deferred to the Appendix.
\section{Functional Data Model For Rotational Curves}\label{scn:FunctionalSO3data}
\subsection{The Rotational Group and Its Isometries}
In the following $\G$ denotes the compact connected Lie group of three-dimensional rotations $G=SO(3)$ which comes with the Lie algebra $\fg=\so(3)=\{A\in \R^{3\times 3}: A^T=-A\}$ of $3\times 3$ skew symmetric matrices. This Lie algebra is a three-dimensional linear subspace of all $3\times 3$ matrices and thus carries the  natural structure of $\R^3$ conveyed by the isomorphism $\iota: \mathbb R^3 \rightarrow \mathfrak{so}(3)$ given by
\begin{equation*} %\label{R3-2-LA:eq}
 \iota(a) = \begin{pmatrix}
  0 & -a_3 & a_2 \\
  a_3 & 0 & -a_1 \\
  -a_2 & a_1 & 0
 \end{pmatrix},\quad \mbox{ for } a = (a_1,a_2,a_3)^T \in \R^3\,.
 \end{equation*}
 This isomorphism exhibits at once the following relation
 \begin{equation}\label{iota-rotation:eq}
  Q \iota(a) Q^T= \iota(Qa)\mbox{ for all }a\in \R^3\mbox{ and }Q\in G\,.
 \end{equation}
We use the scalar product $\langle A,B \rangle=\tr\big(AB^T\big)/2$, which induces the rescaled Frobenius norm $\Vert A\Vert = \sqrt{\tr(AA^T)/2}$,  on all matrix spaces considered in this paper. On all subspaces it induces the \emph{extrinsic metric}, cf. \cite{Bhattacharya2003}. Moreover, we denote with $I_{3\times 3}$ the unit matrix. As usual, $A\mapsto \Exp(A)$ denotes the matrix exponential which is identical to the Lie exponential and gives a surjection $\fg \to G$. Due to skew symmetry, the following \emph{Rodriguez formula} holds
\begin{equation}\label{eq:rodriguez}
 \Exp(A)  = \sum_{j=0}^\infty \frac{A^j}{j!}
	  = I_{3\times 3} + \frac{\sin(\Vert A \Vert)}{\Vert A \Vert}\,A + \frac{1-\cos(\Vert A \Vert)}{\Vert A \Vert^2}\,A^2\,.
\end{equation}
% known as the \emph{Rodriguez formula}.
This yields that the Lie exponential is bijective on $\mathcal{B}_\pi(0) = \{A\in \fg: \|A\| < \pi\}$.  For a detailed discussion, see \cite[p. 121]{Chirikjian2000}.
% more precisely,
% \begin{equation}\label{eq:period-Exp}
% \Exp(A) = \Exp\big(A+2\pi k A/\Vert A \Vert\big)\mbox{ for all $A \in \fg$ and $k\in \Z$.}
% \end{equation}
%  Hence, an inverse of the Lie exponential can be uniquely defined, mapping
% \begin{equation}
%  \G_0 := \{Q\in \G: \tr(Q) > -1\}\,,
% \end{equation}
% bijectively onto
% \begin{equation}
%  V_0 = \{A \in \fg: \Vert A \Vert\leq \pi\}\,.
% \end{equation}
% This mapping is usually denoted as $\Log$. For $g\in \G_0$, we have
% \begin{eqnarray}\label{eq:log-map}
%  \Log(Q) &=&\frac{1}{2{\rm sinc}(\Theta(Q))} \left( Q - Q^T \right),\mbox{ where } 
% 	  \Theta(Q) = {\rm acos}\left( \frac{{\tr}(Q)-1}{2} \right)\,.
% \end{eqnarray}
% For a more detailed discussion, see \cite[p. 121]{Chirikjian2000}.

The scalar product on the Lie algebra is invariant under left and right multiplication from $\G$ i.e., $\langle A,B \rangle = \langle PAQ,PBQ \rangle$ for all $P,Q\in \G$ and it induces a bi-invariant Riemannian metric on all tangent spaces $\mathcal{T}_P\G$ $(P\in \G$) via
\begin{equation*}
 \langle v,w\rangle_{P} = \Big\langle \big(dL_{P}\big)^{-1} v, \big(dL_{P}\big)^{-1}w\Big\rangle\,,
\end{equation*}
where $v,w\in\mathcal{T}_P\G$, $P\in \G$ and $L_P:\G\rightarrow \G, Q\mapsto L(Q)= PQ$ is \emph{left multiplication} by $P$. 

The induced metric $d_{SO(3)}$ on $SO(3)$ can be simply computed from the extrinsic metric
\begin{align}\label{intrinsic-extrinsic-metrix:eq}
 d_{SO(3)}(P,Q)&=2\arcsin \frac{\|P-Q\|}{2}\,,
\end{align}
e.g. \citet[Section 2.2]{Stanfill2013}, who have an additional factor $1/\sqrt{2}$ in the fraction that is absorbed in our definition of the norm.

The space of isometries with respect to this metric will be denoted with $\mathcal{I}(\G)$ and its connected component containing the identity map with $\Inull$. It is well known that both of these groups are Lie groups (see \citet{MyersSteenrood1939}). %Moreover, an application of Theorem 4.1(i) in \citet{Helgason1962} yields the following lemma.

\begin{Lem}\label{lem:isometry-group} $\Inull = \G\times \G$ which acts via $(P,Q):G\to G,~R\mapsto PRQ$. \end{Lem}

\subsection{Data Space and Group Actions}

Recalling that $\G=SO(3)$ we introduce the \emph{data space}
\begin{equation*}
 \X = \mathcal{C}^1\big([0,1],\G\big)\,.
\end{equation*}
On $\X$ we have the action of isometries $\Inull = \G\times \G$ in the connectivity component of the identity -- the analog of orientation preserving Euclidean motions in a Euclidean geometry -- given by
\begin{equation*}
  \big(t\mapsto \gamma(t)\big) \stackrel{(P,Q)}{\to} \big(t\mapsto P\gamma(t)Q\big)
\end{equation*}
for all $(P,Q)\in \G\times \G$ and $\gamma \in \X$, as well as the action of the group (by composition) of time reparametrizations
\begin{equation*} %\label{eq:DiffSpace}
 \Diff = \Big\{ \phi \in \mathcal{C}^\infty\big([0,1],[0,1]\big)\,\big\vert~ \phi'(t) > 0~\text{ for all } t\in (0,1)\Big\}
\end{equation*}
given by
\begin{equation*}
  \big(t\mapsto \gamma(t)\big) \stackrel{\phi}{\to} \big(t\mapsto \gamma\circ \phi (t)\big)\,.
  \end{equation*}
for all $\phi \in \Diff$ and $\gamma \in \X$. Of course, both group actions commute and we can define the direct product group
$$\Inull \times \Diff\,.$$
In the following, when we speak of the \emph{similarity group} $\Shape$ acting, we either refer to the spatial action by $\Shape = \Inull$, or the temporal action by $\Shape = \Diff$, or the joint action  by $\Shape=\Inull \times \Diff$. Any of the three group actions will be denoted by $G \stackrel{g}{\to} G, ~P \mapsto g.P$ for every $g\in\Shape$.

\subsection{Gaussian Perturbation Models}\label{subscn:rGPmodels}

\begin{Def}\label{definition:rGP}
      We say that a random curve $\gamma\in \X$ follows a \emph{Gaussian perturbation} (GP) around a \emph{center curve} $\gamma_0\in \X$, if there is a $\R^3$-valued zero-mean Gaussian processes $A_t$ with a.s. $\cC^1$ paths, such that
      \begin{equation}\label{eqdef:rGP}
	    \gamma(t)  = \gamma_0(t)\Exp\big(\iota\circ A_t\big)\,,
      \end{equation}
      for all $t\in [0,1]$. The Gaussian process $A_t$ will be called the \emph{generating process}. % of the GP.
\end{Def}

From the fact that, 
\begin{equation}\label{exp-commutator:eq}
    Q\Exp(A)Q^T = \Exp(QAQ^T)\mbox{ for all }Q \in \G\mbox{ and }A\in \fg
\end{equation}
we obtain at once the following Theorem.
% The family of GP models has the useful property that it is invariant under the actions of $\Inull$ and $\Diff$. 

\begin{Th}[Invariance of GP Models under Spatial and Temporal Group Actions]\label{thm:rGPEqiv}
    Let $(P,Q) \in \Inull$, $\phi \in \Diff$ and let $\gamma\in \X$ follow a GP around a center curve $\gamma_0\in \X$. Then $P(\gamma\circ \phi)Q$ follows a GP around the center curve $P (\gamma_0\circ \phi)Q$.
\end{Th}

Since $\G$ is non-commutative, GP models as introduced above are not the only canonical generalization of the models \eqref{eq:RDfunctionalModel}. Instead of the perturbation model \eqref{eqdef:rGP}, which is a detailed version of (\ref{eq:GPModel}), involving an a.s. $\mathcal{C}^1$-Gaussian process from the right, one could also consider perturbation models involving an a.s. $\mathcal{C}^1$-Gaussian process $B_t$ from the left, or even from both sides, involving a.s. $\mathcal{C}^1$-Gaussian processes $C_t$ and $D_t$ i.e.,   
\begin{align}
  \label{eqdef:lGP}		\eta(t)  			&= \Exp\big(\iota \circ B_t\big)\, \eta_0(t)\\
  \label{eqdef:fullGP} 	\delta(t)  			&= \Exp\big(\iota \circ C_t\big)\, \delta_0(t)\,\Exp\big(\iota \circ D_t\big)
\end{align}
for all $t\in [0,1]$. Again, invoking (\ref{exp-commutator:eq}) the following Theorem follows readily. 
% Theorem \ref{thm:RightLeftEquivalence} shows that that \eqref{eqdef:rGP} and  \eqref{eqdef:lGP} are equivalent, whereas Theorem \ref{thm:FullAsymptoticEquivalence} establishes that \eqref{eqdef:fullGP} is approximately equivalent to \eqref{eqdef:rGP} and  \eqref{eqdef:lGP} for concentrated errors.
\begin{Th}\label{thm:RightLeftEquivalence}
    Any GP model \eqref{eqdef:rGP} can be rewritten into a left Gaussian perturbation model \eqref{eqdef:lGP} with the same center curve $\gamma_0\in\X$ i.e., for any a.s. $\cC^1$ Gaussian process $A_t$ there exists an a.s. $\cC^1$  Gaussian process $B_t$ such that
    \begin{equation*}
      \gamma_0(t) \Exp\big(\iota \circ A_t\big) = \Exp\big(\iota \circ B_t\big) \gamma_0(t)
    \end{equation*}
    and vice versa.
\end{Th}

% The proof of the following Theorem is more involved and deferred to the Appendix.
\begin{Th} \label{thm:FullAsymptoticEquivalence}
Let $\sigma\rightarrow 0$ be a concentration parameter. Consider a both-sided Gaussian perturbation $\delta(t)$ around a center curve $\delta_0(t)$ given by (\ref{eqdef:fullGP})
% \begin{equation*}
%   \delta(t)  = \Exp\!\left(\iota \circ C_{t}\right) \delta_0(t)\,\Exp\!\left(\iota \circ D_{t}\right)
% \end{equation*}
with $\max_{t\in [0,1]}\Vert C_{t} \Vert = \mathcal{O}_p(\sigma)$ and $\max_{t\in [0,1]}\Vert D_{t} \Vert = \mathcal{O}_p(\sigma)$ for all $t\in [0,1]$. Then $\delta(t)$ can be rewritten into a right Gaussian perturbation i.e.,
\begin{equation*}%\label{eq:ApproxIdentity}
  \delta(t)  =  \delta_0(t)\,\Exp\!\left(\iota\circ A_t + \iota\circ \tilde A_t\right)
\end{equation*}
with an a.s. $\cC^1$ zero-mean Gaussian process $ A_t = \delta_0(t)^TC_t + D_t $ and a suitable zero-mean a.s. $\cC^1$ process $\tilde A_t$ satisfying $\max_{t\in [0,1]}\big\Vert \tilde A_{t} \big\Vert=\mathcal{O}_p\big(\sigma^2\big)$.
\end{Th}      
\begin{Rm} Of course, every right Gaussian perturbation \eqref{eqdef:rGP} or left Gaussian perturbation \eqref{eqdef:rGP} is also a both sided Gaussian perturbation \eqref{eqdef:fullGP}, since the deterministic process $C_t=0$ or $D_t=0$ for all $t\in [0,1]$ is a Gaussian process by definition.
\end{Rm}

% % % % % % % % % % % % % % % % % % % % % % % % % % % Estimation of CC % % % % % % % % % % % % % % % % % % % % % % % % % % %
\subsection{Asymptotic Properties of Center Curve Estimates}\label{subscn:EstimationCenterCurve}

For the entire paper we consider the canonical embedding $\G=SO(3)\hookrightarrow \R^{3\times 3}$. For random curves $\gamma_1,\ldots,\gamma_N,\gamma\in \X$ denote with
\begin{equation*}
 \bar \gamma_N(t)=N^{-1}\sum_{n=1}^N \gamma_n(t) \in \R^{3\times3}\,,
\end{equation*}
the usual Euclidean average curve and by $\E\big[\gamma(t)\big]\in \R^{3\times 3}$ the usual expected curve, if existent. %Both are also called \emph{Euclidean mean} curves.
\begin{Def}\label{definition:PEM}
For random  curves $\gamma_1,\ldots,\gamma_N,\gamma\in \X$, every curve $\mu: [0,1] \to \G$ is called a \emph{pointwise extrinsic Fr\'echet mean} (PEM) curve, if 
\begin{align*}
  \mu(t) &\in \argmin_{\mu\in \G}\E\!\left[\Vert \mu - \gamma(t)\Vert^2\right], ~&&\text{this is a \emph{population} PEM curve, or }\\
  \mu(t) &\in \hat E_N(t) = \argmin_{\mu\in \G}\frac{1}{N}\sum_{n=1}^N  \Vert \mu - \gamma_n(t) \Vert^2,&&\text{this is a \emph{sample} PEM curve}\,,
\end{align*}
respectively. 
\end{Def}

Due to compactness of $\G$ there are always PEM curves. Verify at once the following equivariance property.

\begin{Th}\label{theorem:EquivariancePEM}
 Every population or sample PEM curve is an equivariant descriptor under the action of $\Inull$ and $\Diff$ i.e. if $t\mapsto \mu(t)$ is a population PEM curve for $\gamma \in \X$ or a sample PEM curve for $\gamma_1,\ldots,\gamma_N \in \X$ and $(P,Q) \in \Inull$, $\phi \in \Diff$, then so is
 \begin{align*}
 t\mapsto P(\mu\circ \phi)(t)Q
 \end{align*}
 a population PEM curve for $P(\gamma\circ \phi)Q$, or a sample PEM curve for $P(\gamma_1\circ \phi)Q,\ldots,P(\gamma_N\circ \phi)Q$, respectively.
\end{Th}

In the following we are concerned with circumstances under which PEM curves are unique,
again members of $\X$, and consistent estimators of a GP model's center curve.
% \begin{enumerate}
%  \item[a)] unique,
%  \item[b)] again members of $\X$ 
%  \item[c)] and consistent estimators of a center curve.
% \end{enumerate}
To this end recall the orthogonal projection
\begin{align*}
 \pr: \R^{3\times 3} &\to SO(3),\quad
 A \mapsto \argmin_{\mu\in \G}\|A-\mu\|^2 = \argmax_{\mu\in \G}\tr(A^T\mu)
\end{align*}
which is well defined if and only if $\rank(A)>1$. In this case it assumes the value
\begin{align*}
 \pr(A) = USV^T
\end{align*}
where $A=UDV^T$ is a singular value decomposition (SVD) with decreasing eigenvalues $d_1\geq d_2\geq d_3\geq 0$, $D=\diag(d_1,d_2,d_3)$, and
 \begin{equation*}
  S = \begin{cases}
       I_{3\times 3} 		& \mbox{if } {\rm det}(UV)=1 \\ 
      {\rm diag}(1,1,-1)        & \mbox{if } {\rm det}(UV)=-1 
      \end{cases}\,,
 \end{equation*}
e.g. \citet[Lemma, p. 377]{Umeyama1991}. % or earlier \citet{Stephens1979}.

\begin{Th}\label{theorem:ComputingPEM}
 Let $\gamma_1,\ldots,\gamma_N$ be a sample of a random curve $\gamma \in \X$. %Let  $\gamma_1,...,\gamma_N,\gamma \in \X$ be random curves. 
 Then the following hold.
 \begin{enumerate}
  \item[(i)] Their population or sample PEM curve is unique if and only if $\rank\Big(\E\big[\gamma(t)\big]\Big)>1$, or $\rank\big(\bar \gamma_N(t)\big) >1$, respectively, for all $t\in[0,1]$, and it is then given by
  \begin{align*}
   t\mapsto \pr\Big(\E\big[\gamma(t)\big]\Big),\mbox{ or } t\mapsto \pr\big(\bar \gamma_N(t)\big)\,,
  \end{align*}
  respectively;
  \item[(ii)] in that case the sample PEM is in $\X$ and the population PEM is in $\X$ if additionally $\|\dot\gamma(t)\| \leq Z_t$ with a process $Z_t$ that is integrable for all $t\in [0,1]$.
 \end{enumerate}
\end{Th}

\begin{Th}[Perturbation Consistency]\label{theorem:PEMisCenterCurve}
  If a random curve $\gamma\in \X$ follows a GP around a center curve $\gamma_0\in \X$, %(see Definition \ref{definition:rGP}), 
  then its PEM is curve unique and it is identical to $\gamma_0$.
\end{Th}
\begin{Rm} Close inspection of the proof of the above theorem shows that perturbation consistency is valid far beyond GP models. More precisely, for every perturbation model 
   $\gamma(t) = \gamma_0(t)\Exp(\iota\circ A_t)$ on $\G$ around any central curve $\gamma_0$, induced by any $\R^3$-valued stochastic process $A_t$ satisfying
  \begin{equation*}
      \E[ A_t\, \sinc{ \Vert A_t \Vert }]=0~~ \text{  and  }~~ \E[\cos\Vert A_t\Vert] >0\,,
  \end{equation*}
  its PEM curve is unique and equal to $\gamma_0$. Moreover, $\det(\E[\gamma(t)])>0$.
\end{Rm} 
\begin{Cor}\label{corollary:PESMconsistency}
    Let $\gamma_1,\ldots,\gamma_N$ be a sample of a random curve $\gamma \in \X$ following a GP model around a center curve $\gamma_0$. Fixing $t\in [0,1]$ and choosing a measurable selection $\hat\mu_N(t)$ of the sample PEM curves at $t$, then $\hat \mu_N(t) \to \gamma_0(t)$ almost surely.
\end{Cor}
This, showing that the center curve can be pointwise consistently estimated, can be improved under a mild additional assumption to hold uniformly.
\begin{Th}\label{theorem:PESMuniformconvergence}
    Let $\gamma_1,\ldots,\gamma_N$ be a sample of a random curve $\gamma \in \X$ following a GP model around a center curve $\gamma_0$ and let $t\mapsto \hat \mu_N(t)$ be a measurable selection of $\hat E_N(t)$ for each time point $t\in [0,1]$. If the generating Gaussian process $A_t$  satisfies
 \begin{equation}\label{eq:uniformAssumption}
    \E\left[ \max_{t\in [0,1]} \Vert \partial_t A_t \Vert \right] <\infty\,,
 \end{equation}
%  Considering a curve $t\mapsto \hat \mu_N(t)$ that is a measurable selection of $\hat E_N(t)$ for each time $t\in I$, the following hold.
% Denoting with $t\mapsto \mu_N(t)$ a pointwise measurable selection of sample PEMs at $t \in [0,1]$, 
then the following hold.
 \begin{enumerate}
  \item[(i)] There is $\Omega'\subset \Omega$ measurable with $\Prb(\Omega')=1$ such that for every $\omega\in \Omega'$ there is $N_\omega\in \mathbb N$ such that for all $N\geq N_\omega$, every $\hat E_N(t)$ has a unique element $\mu_N(t)$, for all $t\in [0,1]$, and $\hat \mu_N \in \X$;
  \item[(ii)] $\max_{t\in [0,1]} \big\Vert \hat\mu_N(t) - \gamma_0(t)\big\Vert \rightarrow 0$ and $\max_{t\in [0,1]} d_{\G}\big( \hat\mu_N(t), \gamma_0(t) \big) \rightarrow 0$ for $N\rightarrow\infty$ \mbox{almost surely.}
%   \item[(iii)] $\max_{t\in [0,1]} d_{\G}\big( \hat\mu_N(t), \gamma_0(t) \big) \rightarrow 0$ for $N\rightarrow\infty$ almost surely.
 \end{enumerate}
%  Note, that in (ii) and (iii) as long as $\hat E_N(t)$ is not the unique point $\hat\mu_N(t)$ the distances are taken between a point and a set. 
\end{Th}
\begin{Cor}\label{corollary:PESMoftenContinuous}
With the notations and assumptions of Theorem \ref{theorem:PESMuniformconvergence} we have
 \begin{equation*}
  \lim_{N\rightarrow\infty} \Prb\big\{ t \mapsto \hat\mu_N(t)\in \X \big\} = 1\,.
 \end{equation*}
\end{Cor}

% % % % % % % % % % % % % % % % % % % % % % % % % % % Gait Similarities % % % % % % % % % % % % % % % % % % % % % % % % % % %
\section{Curve Registration}\label{scn:EsimtationGaitSimilarities}

With various applications in mind, for instance only performing spatial registration when time warping is considered as part of the signal, or temporal registration between samples but not within, when within sample time warping is relevant for the signal, in this section we consider both registration issues separately, with the option of arbitrary combinations.
    
We begin with a generic loss function for general Lie groups that is invariant under both group actions, which we will use for temporal registration. As we aim at keeping our methodology open for various combinations of group actions we do not aim at giving the corresponding quotient a Riemannian structure, as in \cite{SuKurtek2014}, say. Rather we then also introduce a loss function specifically tailored to the fact that the unit quaternions represent a double cover of $SO(3)$ allowing to easily perform spatial registration.

\subsection{Generic Loss Functions}\label{scn:generic-loss-fcn}

    For curves $\gamma, \eta: [0,1] \to \RR^D$ it is customary to consider the \emph{total variation} (TV) loss
    $$ (\gamma,\eta) \mapsto \int_{0}^1 \|\dot\gamma(t) - \dot\eta(t)\|\,dt$$
    if existent, which is invariant under temporal reparametrizations and spatial isometries. If $\RR^D$ is viewed as an additive Lie group, this loss has the following three canonical generalizations.
    
  \begin{Def}\label{def:ILLS}
    The \emph{intrinsic length losses} %(ILLs) 
    on $\X$ are given by 
    \begin{align*}%\label{eq:Lengthloss}
	\delta_{I,1}(\gamma,\eta) &= {\rm length}(\gamma\eta^{-1})~ ~ \text{ and }  ~ ~ \delta_{I,2}(\gamma,\eta) = {\rm length}(\gamma^{-1}\eta)\,,
    \end{align*} 
    for $\gamma,\eta \in \X$. Here the length is taken with respect to the bi-invariant metric on $\G$,
       \begin{align*} 
       {\rm length}(\gamma) = \int_0^1 \|\dot \gamma(t)\|\,dt 
       \end{align*}
     and,
     in order to have a loss without choosing where to put the inverse, define
     \begin{equation*}
      \delta_{I}(\gamma,\eta) = \tfrac{1}{2}\big(\delta_{I,1}(\gamma,\eta) + \delta_{I,2}(\gamma,\eta)\big)\,.
     \end{equation*}
  \end{Def}
    
    Each of these is again invariant under temporal and spatial reparametrizing, and more. 
    
     \begin{Th}\label{theorem:InvariantIntrinsicLoss}
	Let us denote with $\delta$ either $\delta_{I,1}$, $\delta_{I,2}$ or $\delta_{I}$. Then the following hold
	\begin{enumerate}
	\item[(i)]   $\delta$ is symmetric.
	\item[(ii)]  $\delta$ is invariant under the action of $\Shape=\Inull\times \Diff$, i.e{.} 
	\begin{equation*}
	 \delta\big( (\psi,\phi).\gamma, (\psi,\phi).\eta \big) = \delta(\gamma, \eta)
	\end{equation*}
	for all $(\psi,\phi)\in \Shape$ and all $\gamma,\eta\in \X$.
	\item[(iii)] Let $P\in \G$ be arbitrary. Then $\delta_{I,1}(P\gamma,\eta)= \delta_{I,1}(\gamma,\eta)$ and $\delta_{I,2}(\gamma P,\eta)= \delta_{I,2}(\gamma,\eta)$ for all $\gamma,\eta\in \X$. 
	\item[(iv)] $\delta_{I,1}(\gamma,\eta)= 0$ if and only if $\eta=\gamma P$ for some $P\in \G$. Similarly, $\delta_{I,2}(\gamma,\eta)= 0$ if and only if $\eta=P\gamma$ for some $P\in \G$.
	\item[(v)]
   Let $\gamma, \eta\in\X$. Then we have the \emph{inverse alignment property}
%    \begin{equation*}
   \begin{eqnarray*}
%       &
      (\phi^*,\psi^*) \in \mathop{\arginf}_{(\phi,\psi)\in \Shape} \delta(\gamma_1,(\phi,\psi).\gamma_2) 
%       \\&\Downarrow \\&
      ~\Rightarrow~
      \big((\phi^*)^{-1},(\psi^*)^{-1}\big) \in \arginf_{(\phi,\psi)\in \Shape} \delta(\gamma_2,(\phi,\psi).\gamma_1)
   \end{eqnarray*}
	\end{enumerate}
   \end{Th}
 Although all three of our intrinsic loss functions fulfill the inverse alignment property, if the infimum is attained, the joint minimization is challenging and we use minimization w.r.t. $\delta$ (i.e. $\delta_{I,1}$, $\delta_{I,2}$ or $\delta_{I}$) only to obtain 
   \begin{equation}\label{eq:TwoSampleSimpleRegistration}
     \phi^*\in \arginf_{\phi\in \Diff} \delta(\gamma_1,\gamma_2\circ\phi)\,,
  \end{equation}
 using a dynamic program, following \citet{Kurtek2011,SrivastavaWu2011,SrivastavaKlassen2011}.
  
%  All three of our intrinsic loss functions fulfill the inverse alignment property with respect to time warping discussed in \citet{Vantini2012} or \citet{Kurtek2011}, if the infimum is attained. An inverse alignment property with respect to spatial alignment is established in Theorem \ref{theorem:Lift}, (iii).
%   \begin{Th}\label{theorem:InverseAlignmentDiff}
%    Let $\gamma, \eta\in\X$. Then we have that
%    \begin{equation*}
%       \phi^* \in \arginf_{\phi\in \Diff} \delta(\gamma_1,\gamma_2\circ\phi) \Rightarrow (\phi^*)^{-1} \in \arginf_{\phi\in \Diff} \delta(\gamma_2,\gamma_1\circ\phi)
%    \end{equation*}
%   \end{Th}
    
%   \paragraph{Temporal registration of two curves.}
%   Following the general ideas presented in \citet{Kurtek2011}, \citet{SrivastavaWu2011} and \citet{SrivastavaKlassen2011} temporal registration of two curves $\gamma_1,\gamma_2 \in \X$ is done by finding $\phi^*\in \Diff$ approximating
%   \begin{equation}\label{eq:TwoSampleSimpleRegistration}
%       \arginf_{\phi\in \Diff} \delta(\gamma_1,\gamma_2\circ\phi)\,.
%   \end{equation}
%   Here again $\delta$ denotes either $\delta_{I,1}$, $\delta_{I,2}$ or $\delta_{I}$. The specific loss is the only choice needed in our approach. In general one should pick $\delta_I$, since it does treat left and right translation of curves equally. 

\subsection{Spatial Alignment} %: Marker Placement Effect}
\paragraph{A unit quaternion point of view} allows for a direct and noniterative spatial alignment given by
\begin{align}\label{SO3-spatial-loss:eq} 
E_{\gamma,\eta} &= \argmin_{P,Q\in SO(3)} L(P\gamma Q,\eta)\,,
\end{align}
for $\gamma,\eta \in X$
with a loss function $L$ induced by a suitable loss function on the unit quaternions below in (\ref{quat-loss:def}) and (\ref{eq:ExtrinsicL2loss}). 

The unit quaternions are given by  $\SSS^3=\{x\in \RR^4: \|x\| = 1\}$ equipped with a multiplicative group structure 
% We start with the observation
% % and effective define a loss function that allows 
% % Let us assume that $\gamma$ and $\eta$ differ only by a spatial isometry from the connectivity component of the identity i.e., there exists $(P,Q)\in\Inull$ such that
% % \begin{equation}\label{equation:spatialAligning}
% %   \eta=P\gamma Q^T.
% % \end{equation}
% % In order to solve equation \eqref{equation:spatialAligning} for $(P,Q)$ given $\gamma$ and $\eta$ we note 
% that any curve $\gamma\in\X$ can be lifted to a continuous curve in $\SSS^3=\{x\in \RR^4: \|x\| = 1\}$, which we view as a the multiplicative group of unit quaternions 
via
\begin{equation*}
 \SSS^3\ni (x_1, x_2, x_3, x_4)^T~ ~ \leftrightarrow ~ ~ x_1 + ix_2 + jx_3 + kx_4
\end{equation*}
with $i^2=j^2=k^2=-1$ and $i\cdot j=-j\cdot i=k$, $k\cdot i=-i\cdot k=j$, $j\cdot k=-k\cdot j=i$ (e.g., \cite{Chirikjian2000}). %, here ``$\cdot$'' denotes the quaternion multiplication. 
Moreover, the map
\begin{equation*}
  \begin{aligned}%\label{eq:S3SO3proj}
  \pi:~\SSS^3 		&\rightarrow \G\\
		    \begin{pmatrix}
				x_1 \\
				x_2 \\  
				x_3 \\
				x_4
		    \end{pmatrix}	&\mapsto \begin{pmatrix}
				1 - 2x_3^2 - 2x_4^2   &  2(x_2x_3 + x_1x_4) &  2(x_2x_4 - x_1x_3) \\
				2(x_2x_3 - x_1x_4)    &  1 - 2x_2^2-2x_4^2   &  2(x_3x_4+x_1x_2)  \\  
				2(x_2x_4+x_1x_3) &  2(x_3x_4-x_1x_2) &  1-2x_2^2-2x_3^2 )
		    \end{pmatrix}
  \end{aligned}
\end{equation*}
is a double (even universal) cover of $\G$ and a smooth surjective group homomorphism with the property $\pi(x) = \pi(-x)$ for all $x\in \SSS^3$ (see \cite{Stuelpnagel1964}). Thus, by the lifting property of covering maps (e.g., \citet[Proposition A.77, p.616]{Lee2013}) any curve $\gamma\in\X$ has exactly two continuous lifts $\tilde\gamma$ in $\SSS^3$, each uniquely determined by the choice of the starting element from $\pi^{-1}\big(\gamma(0)\big)$. 

In consequence, every continuous loss function $\wL : \mathcal{C}\big([0,1],\SSS^3\big)\times\mathcal{C}\big([0,1],\SSS^3\big)\rightarrow\R_{\geq0}$ invariant under common sign changes induces a continuous loss function $L:\X\times\X\rightarrow\R_{\geq0}$ via
\begin{align}\label{quat-loss:def}
 L(\gamma,\eta ) = \min\left\{\wL(\tilde\gamma,\tilde\eta),\wL(-\tilde\gamma,\tilde\eta)\right\}
\end{align}
where $\tilde\gamma$ and $\tilde\eta$ are arbitrary continuous lifts of $\gamma$ and $\eta$.

The following setup now allows to compare minimizers of (\ref{SO3-spatial-loss:eq}) with suitably defined minimizers for $\wL$ because the action of $\Inull$ on $\X$ lifts to the canonical left action of $SO(4)$ on continuously lifted curves $\tilde\gamma \in\mathcal{C}\big([0,1],\SSS^3\big)$ as in (\ref{eq:IsoLift}) below. The following is well known,
where (i) and (ii) below are detailed in \citet{Mebius2005} and (iii) is a consequence of (i) because the quaternion multiplication is associative. \begin{Lem}\label{lemma:Quat2Rot}
 Let $ p =(p_1,p_2,p_3,p_4)^T\in \SSS^3$ and $ q =(q_1,q_2,q_3,q_4)^T\in \SSS^3$ be arbitrary unit quaternions. Also consider arbitrary $R\in SO(4)$. Then the following hold.
 \begin{enumerate}
  \item[(i)]
 There are unique
  \begin{equation*}
  R^l_{p} = \begin{pmatrix}
                p_1  & -p_2  & -p_3 & -p_4 \\
		p_2  &  p_1  & -p_4 &  p_3 \\  
		p_3  &  p_4  &  p_1 & -p_2 \\
		p_4  & -p_3  &  p_2 &  p_1 \\
        \end{pmatrix}\in SO(4)\,,~~
  R^r_{q} = \begin{pmatrix}
                q_1  & -q_2  & -q_3 & -q_4 \\
		q_2  &  q_1  &  q_4 & -q_3 \\  
		q_3  & -q_4  &  q_1 &  q_2 \\
		q_4  &  q_3  & -q_2 &  q_1 \\
        \end{pmatrix}\,\in SO(4)
  \end{equation*}
 such that $p\cdot v = R^l_{p}v$ and $v\cdot q = R^r_{q}v$ for all $v\in \SSS^3$.
 \item[(ii)] There are, unique up to common sign change, unit quaternions $p_R,q_R\in \SSS^3$ such that $Rv = p_R\cdot v\cdot q_R$ for all $v\in \SSS^3$.
 \item[(iii)]
 With the notation from (i), there is a smooth surjective group homomorphism
 \begin{equation*}
    \begin{aligned}\ %label{eq:S3S3SO4proj}
	\pi_{SO(4)}: \SSS^3\times \SSS^3	&\rightarrow SO(4),\quad
	(p,q) 				&\mapsto  R^l_{p} R^r_{q} = R^r_{q} R^l_{p} 
    \end{aligned}
 \end{equation*}
 with $\pi_{SO(4)}(p,q) = \pi_{SO(4)}(-p,-q)$.
% \item[(iv)] For any $R\in SO(4)$ we have that $\Pi(R)=\Pi(-R)$ and $\Pi(R^{-1})=\Pi(R)^{-1}$.
  \end{enumerate}
\end{Lem}

\begin{Def}\label{def:IndepProj}
With an arbitrary right inverse $r_{SO(4)}$ of the double cover $\pi_{SO(4)}\!:\SSS^3\times \SSS^3\rightarrow SO(4)$ from Lemma \ref{lemma:Quat2Rot}, define $\Pi=(\pi,\pi)\circ r_{SO(4)} : SO(4) \to \Inull$. 
% \begin{equation*}
% \Pi:SO(4) \xrightarrow {r_{SO(4)}}  \SSS^3\times \SSS^3 \xrightarrow{ \pi\times\pi}  \G \times \G\,.
% \end{equation*}
\end{Def}
The above map is well defined due to $\pi(x) = \pi(-x)$. 
% Note that the above map $\Pi: SO(4)\to \G\times \G$ is well defined because by $\pi(x) = \pi(-x)$ for all $x\in \SSS^3$, 
% every choice of a right inverse $r_{SO(4)}$ gives the same result for $\Pi$. 
% Moreover, we have the following {\bf drop? proofs don't need it anymore}.
% \begin{Lem}\label{lemma:IndepProj}
% For any $R\in SO(4)$ we have that $\Pi(R)=\Pi(-R)$ and $\Pi(R^{-1})=\Pi(R)^{-1}$.
% \end{Lem}

% 
% Using the above lifts the action of $\Inull$ on $\X$ lifts to the canonical left action of $SO(4)$ on continuously lifted curves $\tilde\gamma \in\mathcal{C}\big(I,\SSS^3\big)$.
\begin{Th}\label{theorem:Lift} 
Consider curves $\gamma,\eta \in \X$ with arbitrary continuous lifts $\tilde\gamma,\tilde\eta$ and let $L$ be defined as in (\ref{quat-loss:def}) with a continuous loss function $\wL$, invariant under common sign changes. Then the following hold.
\begin{itemize}
 \item[(i)] For  $(P,Q)\in \Inull $ and arbitrary continuous lift  $\widetilde{P\gamma Q}$ of $P\gamma Q$,
\begin{itemize}
 \item[(a)] 
there is a unique $R\in SO(4)$ with the property
 \begin{equation}\label{eq:IsoLift}
  \widetilde{P\gamma Q}(t)=R \tilde\gamma(t)\mbox{ for all }t\in[0,1]\,,
 \end{equation}
 \item[(b)]
 this satisfies %If $R\in SO(4)$ is the unique element satisfying (i), %equation \eqref{eq:IsoLift} in Theorem \ref{theorem:Lift}. Then $\Pi(R)=(P,Q^T)$.
%  then 
 $\Pi(R)=(P,Q)$.
 \end{itemize}
 \item[(ii)] 
% If $L$ is defined as in (\ref{quat-loss:def}) with a symmetric $SO(4)$-invariant loss function $\wL$ then
%  \begin{align*}
$  \Pi\Bigg( \mathop{\argmin}_{R \in SO(4)} \wL\big( R\tilde\gamma, \tilde\eta \big) \Bigg)= E_{\gamma,\eta}=\argmin_{P,Q\in SO(3)} L(P\gamma Q,\eta)\,.$
 \item[(iii)] If $\wL$ is additionally symmetric and $SO(4)$-invariant then the inverse alignment property holds
  \begin{align*}(P,Q) \in E_{\gamma,\eta} \Leftrightarrow (P^{-1},Q^{-1}) \in  E_{P\gamma Q,\eta}  \end{align*}

\end{itemize}

\end{Th}

\paragraph{Asymptotic spatial alignment for two samples.}
Consider the $SO(4)$ invariant loss
% 
% Assume we have two sessions $\chi_1=(\gamma_1,...,\gamma_{N})$, $N\in \N$, consisting of realizations of an rGP $\gamma$ with center curve $\gamma_0$ and $\chi_2=(\eta_1,..., \eta_{M})$, $M\in \N$, consisting of realizations of an rGP $\eta$ with center curve $\eta_0$.  We additionally assume that the PESMs $t\mapsto\mu_N(\chi_1,t)=\hat \gamma_N(t)$ and $t\mapsto\mu_N(\chi_1,t)=\hat \eta_M(t)$ are elements in $\X$, which is usually the case in biomechanical applications, since the data is concentrated around the center curve. Theoretically, this assumption is asymptotically satisfied by Corollary \ref{corollary:PESMoftenContinuous}. In order to estimate the effect of marker placements between $\gamma$ and $\eta$, we propose to use the estimator
% \begin{equation}\label{eq:MRestimCurve}
%  \big(\hat P, \hat Q\big) \in \Pi\Bigg(\argmin_{R \in SO(4)} \int_{0}^1 \Big\Vert R \widetilde{\hat \gamma_N(t)}- \widetilde{\hat \eta_M(t)} \Big\Vert^2\,dt \Bigg)\,,
% \end{equation}
% which is procedure \eqref{eq:Restimator} applied to the PESMs of $\chi_1$ and $\chi_2$ and the $SO(4)$-invariant loss
\begin{equation}
\begin{aligned}\label{eq:ExtrinsicL2loss}
   \wL_2:\mathcal{C}\big([0,1],\SSS^3\big)\times\mathcal{C}\big([0,1],\SSS^3\big)~~	&\rightarrow	 ~~~~~~~~	\R_{\geq0}\\
					  (\tilde\gamma,\tilde\eta) ~~~~~~~~~~~	&\mapsto 	~~	\int_{0}^1 \Vert \tilde\gamma(t)-\tilde\eta(t) \Vert^2\,dt\,,
\end{aligned}
\end{equation}
which, via (\ref{quat-loss:def}), induces the loss $L_2$ on $\X$.

Then,  as in before Section \ref{subscn:EstimationCenterCurve}, the minimization of $\wL_2(R\tilde\gamma,\tilde\eta)$ over $R\in SO(4)$ is equivalent to maximizing $\tr(RH_{\tilde\gamma,\tilde\eta})$ with
\begin{align*} H_{\tilde\gamma,\tilde\eta} = \int_0^1 \tilde\eta(t)\tilde\gamma(t)^T \,dt\,.
\end{align*}
This is, again alluding to  \citet[Lemma, p. 377]{Umeyama1991}, solved by a SVD, from which, in conjunction with Theorem \ref{theorem:Lift} we obtain at once the following. 
% 
% The main advantage of lifting the minimization problem to $SO(4)$ instead of $SO(3)\times SO(3)$ is that based on work in spherical regression by \citet{Stephens1979} (see also \citet{Umeyama1991}) the minimization is solvable by a singular value decomposition and uniqueness conditions can be stated easily. Moreover, this estimator turns out to be strongly consistent.
\begin{Th}\label{theorem:computingMPestim}
 Let $\gamma, \eta\in \X$ with arbitrary continuous lifts $\tilde \gamma, \tilde \eta$ to $\SSS^3$ and let $L_2$ be the loss on $\X$ induced by $\wL_2$ through (\ref{quat-loss:def}). Then, all elements in $E_{\gamma,\eta}$ are of the form $\Pi\big( U S V^T \big)$, where $U,V\in O(4)$, together with $D={\rm diag}(d_1,...,d_4)$, $d_1\geq...\geq d_4\geq0$, are from a SVD, $H_{\tilde\gamma,\tilde\eta}=UDV^T$, and
 \begin{equation*}
  S = \begin{cases}
       I_{4\times4} 		& \mbox{if } {\rm det}(UV)=1 \\ 
      {\rm diag}(1,1,1,-1)        & \mbox{if } {\rm det}(UV)=-1 
      \end{cases}\,.
 \end{equation*}
 If additionally ${\rm rank}(H_{\tilde\gamma,\tilde\eta}) >2$, then $(P^*,Q^*) \in E_{\gamma,\eta}$ is unique. 
\end{Th}

\begin{Rm}\label{H:rm}
Note that ${\rm rank}(H_{\tilde\gamma,\tilde\eta})$ is independent of the particular choice of continuous lifts. Moreover in view of our two sample test, under the null hypothesis $\eta_0=P\gamma_0Q$ for some $P,Q\in SO(3)$, we have
$H_{\tilde\gamma_0,\tilde\eta_0} = R$ with suitable $R\in SO(4)$ due (\ref{eq:IsoLift}), which has maximal rank.\end{Rm}

\begin{Th}\label{theorem:GeneralConsistency}
 Assume $\gamma_1,..., \gamma_N\iid \gamma$, $N\in \mathbb{N}$ and $\eta_1, ..., \eta_{M}\iid \eta$, $ M\in \mathbb{N}$ are samples from GP models with center curves $\gamma_0$ and $\eta_0$ respectively, such that that their generating Gaussian processes fulfill Assumption \eqref{eq:uniformAssumption} of Theorem \ref{theorem:PESMuniformconvergence}. Further suppose that $\hat\gamma_N$ and $\hat\eta_M$, respectively, are pointwise measurable selections of the sample PEM curves. Denoting with $\tilde \gamma_0,\widetilde{\hat \gamma_{N}}, \tilde \eta_0,\widetilde{\hat \eta_{M}}$ arbitrary continuous lifts of $\gamma_0,\hat\gamma_N,\eta_0,\hat\eta_M$, under the additional assumption of $\rank(H_{\tilde\gamma_0,\tilde \eta_0}) >2$, we have for every measurable selection $(\hat P_{N, M},\hat Q_{N, M}\Big)\in E_{\hat\gamma_N,\hat\eta_M}$ that 
 \begin{equation*}
      \Big(\hat P_{N, M},\hat Q_{N, M}\Big)\xrightarrow{N,M\rightarrow\infty} \big(P^*,Q^*\big) ~ ~ a.s.\,,
 \end{equation*}
 where $(P^*,Q^*)$ is the unique element in $E_{\gamma_0,\eta_0}$.
%  Here, neither $\big(\hat P_{N, M},\hat Q_{N, M}\big)$ nor $\big(P^*,Q^*\big)$ depend on the particular choices of the continuous lifts.
\end{Th}

\begin{Rm}
The only ingredient of the GP model needed in the proof of Theorem \ref{theorem:GeneralConsistency} is the uniform convergence of sample PEM curves to a unique population PEM. Moreover, in case of $\eta_0 = P\gamma_0Q$ for $P,Q\in SO(3)$, due to Remark \ref{H:rm}, the condition $\rank(H_{\tilde\gamma_0,\tilde \eta_0}) >2$ is always fulfilled.
% Thus, any process fulfilling a Theorem similar to $\ref{theorem:PESMuniformconvergence}$ yields a consistent estimator for the aligning element $(P^*,Q^*)\in\Inull$.
\end{Rm}

% \begin{Cor}\label{corollary:GeneralConsistency}
%  Additionally to the assumptions and notations of Theorem \ref{theorem:GeneralConsistency} assume that $\eta_0=P\gamma_0 Q^T$ for $P,Q\in \G$, i.e{.} $\gamma_0$ and $\eta_0$ only differ by marker placement. Then we have that
%  \begin{equation*}
%       \Pi\Bigg(\argmin_{R\in SO(4)} \int_0^1 \left\Vert R\widetilde{\hat\gamma_N}(t) - \widetilde{\hat\eta_{M}}(t)\right\Vert^2 \Bigg)\ni\Big(\hat P_{N, M},\hat Q_{N, M}\Big)\xrightarrow{N,M\rightarrow\infty} \big(P,Q^T\big) ~ ~ a.s.
%  \end{equation*}
% \end{Cor}

% {\bf ---------------- begin not much touched yet ------------}

% % % % % % % % % % % % % % % % % % % % % % % % % % % Testing % % % % % % % % % % % % % % % % % % % % % % % % % % %
\section{Two Sample Permutation Tests}
% For Equal Center Curves}
\label{scn:Testing}
% In order to test the null hypothesis of equality of center curves of two samples from GP models we propose the following permutation test. On the samples we allow by setting either 
% $\Shape = 1$ (no group action), or $\Shape = \Inull$ (only spatial isometries acting), or $\Shape = \Diff$ (only time warping actions), or
%  $\Shape = \Inull\times \Diff$ (jointly spatial isometries and time warping actions). 

% \begin{itemize}
%  \item $\Shape = 1$: no group action, or
%  \item $\Shape = \Inull$: only spatial isometries acting, or
%  \item $\Shape = \Diff$: only time warping actions, or
%  \item $\Shape = \Inull\times \Diff$: jointly spatial isometries and time warping actions. 
% \end{itemize}
Among others, the following permutation tests are based on intrinsic loss functions from Definition \ref{def:ILLS} such that  $\delta$ denotes either $\delta_{I,1}$, $\delta_{I,2}$ or $\delta_{I}$.

Assume for the following that $\chi_1=\{\gamma_1,...,\gamma_N\}$, $N\in \N$, and $\chi_2=\{\eta_1,...,\eta_M\}$, $M\in \N$, denote samples of GP models with center curves $\gamma_0$ and $\eta_0$ respectively. By virtue of Corollary \ref{corollary:PESMoftenContinuous} we may assume that all of the following sample PEMs are unique elements in $\X$. We want to test
\begin{equation}\label{eq:NullHypoRGP}
 H_0:~~\exists g\in \Shape:~ \eta_0=g.\gamma_0~~~ ~vs.~~~ ~H_1:~~\forall g\in \Shape:~ \eta_0\neq g.\gamma_0\,.
\end{equation}
For a given significance level $\alpha\in(0,1)$ we propose the following tests, first for  $\Shape = 1$ (no group action), and then for $\Shape = \Inull\times \Diff$ (joint group action). Tests for other selections of the groups acting are  similar.

\begin{Test}[No Group Action]\label{test:ILLP}
% \begin{enumerate}
%  \item 
 Let $\chi = \chi_1\cup \chi_2$ be the pooled sample. Denote with $\chi_1^{(l)}$, $l\in\big\{1,...,\binom{N+M}{N}\big\}$, all possible choices with $N$ elements from $\chi$ and with $\chi_2^{(l)}$ its complement in $\chi$ with $M$ elements, such that $\chi_1^{(1)}=\chi_1$ and $\chi_2^{(1)}=\chi_2$. For each $l$ compute the sample PEMs $t\mapsto\hat\gamma_{N,l}(t)$ of $\chi_1^{(l)}$ and $t\mapsto\hat\eta_{M,l}(t)$ of $\chi_2^{(l)}$, respectively, and compute $d_l=\delta\big(\hat\gamma_{N,l},\hat\eta_{M,l}\big)$.
%  \item 
 
 Setting $r = \#\{ l \,\vert~ d_l > d_1\} $, reject the null hypothesis of equality of $\gamma_0$ and $\eta_0$ at significance level $\alpha$ if $r/\binom{N+M}{N} < \alpha$.
% \end{enumerate}
\end{Test}

In order to include group actions, we modify Test \ref{test:ILLP} by first estimating the group element $g$ via the following sample alignment procedure.

\begin{Proc}[Sample Alignment:  $\chi_1$ to $\chi_2$]\label{Proc1} ~
\begin{enumerate}
\item Compute the PESMs $\hat\gamma_N$ and $\hat \eta_M$ of $\chi_1$ and $\chi_2$, respectively by Theorem \ref{theorem:ComputingPEM}.
\item Compute the spatial estimator $\hat\psi=\big(\hat P_\psi, \hat Q_\psi\big)\in E_{\hat\gamma_N,\hat\eta_M}$ from Theorem \ref{theorem:computingMPestim}.
\item Compute the temporal estimator $\hat\phi$ from \eqref{eq:TwoSampleSimpleRegistration} for the curves $\big(\hat\psi,id_{[0,1]}\big).\hat\gamma_N$ and $\hat \eta_M$.
\item Define the new sample $\chi_1 = \big\{\big(\hat\psi,\hat\phi\big).\gamma_1,...,\big(\hat\psi,\hat\phi\big).\gamma_N\big\}$
\item Plug this into %Replace $\hat\gamma_N$ in 
Step 1 %to obtain new $\big(\hat\psi,\hat\phi\big).\hat\gamma_N$ 
and repeat Steps 2-5 until convergence of $\hat\psi$ and $\hat\phi$.
\item Return $\chi_1$ as $\chi^{aligned}_1$.
\end{enumerate}
\end{Proc}

\begin{Test}[Joint Group Action with Preregistration]\label{test:PILLP}
Apply Test \ref{test:ILLP} to the aligned samples $\big(\hat\psi,\hat\phi\big).\chi_1$ and $\chi_2$ where $(\hat \psi,\hat\phi)$ has been determined by the Sample Alignment Procedure \ref{Proc1}.
\end{Test}

Since extrinsic distances and  the  distance $\delta$ are minimized in the alignment, the $\delta$-distance
% Although also extrinsic distances are minimized in the aligment, the  distance $\delta$ 
between $\chi_1^{aligned}$ and $\chi_2$ tends to be smaller then that of permuted subsamples, giving Test \ref{test:PILLP} a level lower than $\alpha$, cf. Table \ref{table:MILLP} (c). In order to keep the level, we perform alignment for each permutation.

% Because in general, $\big(\hat\psi,\hat\phi\big)\neq (\psi,\phi)$, the level of this test is not necessarily $\alpha$, in fact for our simulations the level is smaller, cf. Table \ref{table:MILLP} (b). In order to keep the level, also under group actions, we propose the following modified version including the above preprocessing in each permutation.
\begin{Test}[Joint Group Action with Continual Registration]\label{test:MILLP} ~
%  Given two sessions $\chi_1=(\gamma_1,...,\gamma_N)$, $N\in \N$, consisting of trials of a distribution with unique PEM $\gamma_0$ and $\chi_2=(\eta_1,...,\eta_M)$, $M\in \N$, consisting of trials of of a distribution with unique PEM $\eta_0$. Let $\alpha\in(0,1)$ be a given significance level.
% \begin{enumerate}
%  \item 
 Let $\chi = \chi_1\cup \chi_2$ be the pooled sample. Denote with $\chi_1^{(l)}$, $l\in\big\{1,...,\binom{N+M}{N}\big\}$, all possible choices with $N$ elements from $\chi$ and with $\chi_2^{(l)}$ its complement with $M$ elements in $\chi$ such that $\chi_1^{(1)}=\chi_1$ and $\chi_2^{(1)}=\chi_2$. 
 
%  For each $l$ compute the sample PEMs $t\mapsto\hat\gamma_{N,l}(t)$ of ${\chi_1^{(l)}}^{aligned}$ and $t\mapsto\hat\eta_{M,l}(t)$ of $\chi_2^{(l)}$, respectively, where ${\chi_1^{(l)}}^{aligned}$ is obtained by applying the Sample Aligment Procedure \ref{Proc1} to ${\chi_1^{(l)}}$ and ${\chi_2^{(l)}}$, and compute $d_l=\delta\big(\hat\gamma_{N,l},\hat\eta_{M,l}\big)$.

 For $i,j\in\{1,2\}$ let $\chi_{ij}^{(l)}$ denote the sample consisting of all trials of $\chi_i^{(l)}$ also belonging to $\chi_j$. For each $l$, $i$ compute the sample PEM $\hat\eta_{i,l}$ of $\chi_{i2}^{(l)}$.

 For each $i,l$ apply the Sample Alignment Procedure \ref{Proc1} to the two samples $\chi_{i1}^{(l)}$ and $\chi_{i2}^{(l)}$ to obtain the sample PEM $\hat\gamma_{i,l}^{aligned}$ of the aligned sample ${\chi_{i1}^{(l)}}^{aligned}$. Compute the sample PEM $\hat\omega_{i,l}$ of the sample $\big\{\hat\eta_{i,l},\hat\gamma_{i,l}^{aligned}\big\}$.

%  For $i,j\in\{1,2\}$ let $\chi_{ij}^{(l)}$ denote the sample consisting of all trials of $\chi_i^{(l)}$ also belonging to $\chi_j$. For each $l$, $i$ compute the sample PEMs $t\mapsto\hat\gamma_{i,l}(t)$ of $\chi_{i1}^{(l)}$ and $t\mapsto\hat\eta_{i,l}(t)$ of $\chi_{i2}^{(l)}$, respectively.
% %  \item 
%  
%  For each $i,l$ apply the Sample Aligment Procedure \ref{Proc1} to the two samples $\hat\gamma_{i,l}$ and $\hat\eta_{i,l}$ to obtain an aligned version $\hat\gamma_{i,l}^{aligned}$. Compute the sample PEM $\hat\omega_{i,l}$ of the pooled sample $\hat\eta_{i,l}$ and $\hat\gamma_{i,l}^{aligned}$.
% %  \item 
 
 Apply the Sample Alignment Procedure \ref{Proc1} to the two single element samples $\hat\omega_{1,l}$ and $\hat\omega_{2,l}$ to obtain an aligned version $\hat\omega_{1,l}^{aligned}$ of $\hat\omega_{1,l}$ and compute $d_k=\delta\big(\hat\omega_{1,l}^{aligned},\hat\omega_{2,l}\big)$.
%  \item 
 
 Letting $r = \#\{ l \,\vert~ d_l > d_1\} $, reject the null hypothesis that there is $g\in \Shape$ such that $g.\gamma_0=\eta_0$ at significance level $\alpha$  if $r/\binom{N+M}{N} < \alpha$.
% \end{enumerate}
\end{Test}

\section{Simulations}\label{scn:Simulations}
\paragraph{GP models.}
The performance of our three tests for testing samples of GP models, with and without action of the group $\Inull$, is studied using the following Gaussian processes
\begin{align*}
 \varepsilon^{1}_{t} &= a_1\sin\left( \tfrac{\pi}{2}t \right) + a_2\cos\left( \tfrac{\pi}{2}t \right) \\ 
 \varepsilon^{2}_{t} &= (\sin(4\pi t) + 1.5)\frac{ \sum_{i=1}^{10} a_i \beta_i(x) }{\sqrt{ \sum_{i=1}^{10} \beta_i(x)}},~\beta_i(x)=e^{-\frac{\left(x-\frac{i-1}{9}\right)^2}{0.2}},~i=1,\ldots,10 
%  \varepsilon^{2}_{t} &= (\sin(4\pi t) + 1.5)\frac{ \sum_{i=1}^{10} a_i e^{-\frac{\left(x-\frac{i-1}{9}\right)^2}{0.2}} }{\sqrt{ \sum_{i=1}^{10} e^{-2\frac{\left(x-\frac{i-1}{9}\right)^2}{0.2}} }}
\end{align*}
and center curves $\gamma_{0}^{\lambda}(t)$, $\lambda\in \R$,
\begin{align*}
 \alpha_x^\lambda\big(t\big) &= 80t^2-80t+20 + \lambda\frac{e^{-\tfrac{1}{2}\big(\tfrac{t-0.5}{0.08}\big)^2}}{0.08\sqrt{2\pi}}-35		\\
 \alpha_y\big(t\big) &= 70t\sin\big(4\pi t^{0.7}\big) + 5			\\
 \alpha_z\big(t\big) &= 10\cos\big(13\pi \big)\,,
\end{align*}
where $\alpha_x^\lambda,\alpha_y, \alpha_z$ are Euler angles of $\gamma_{0}^{\lambda}(t)$ measured  in degrees. Then, five GP models are defined by
\begin{equation}\label{eq:PowerTestProcesses}
    \begin{aligned}
    \gamma^0_{A} &= \gamma_{0}^{0}(t)\Exp\Big(\iota\big((\varepsilon^{1}_{1,t},\varepsilon^{1}_{2,t},\varepsilon^{1}_{3,t})\big)\Big)		\\
    \gamma^\lambda_{B} &= \gamma_{0}^{\lambda}(t)\Exp\left(\iota\left(\begin{pmatrix}
		1 		& 0 		& 0 \\
		\frac{1}{2} 	& \frac{1}{2} 	& 0 \\
		\frac{1}{\sqrt{3}} 		& \frac{1}{\sqrt{3}} 		& \frac{1}{\sqrt{3}}
	  \end{pmatrix}(\varepsilon^{2}_{1,t},\varepsilon^{2}_{2,t},\varepsilon^{2}_{3,t})^T\right)\right) ~ ~ \text{ for } \lambda\in\{0.5,1,2,2.5\}\,.
    \end{aligned}
\end{equation}
where $\varepsilon^{j}_{i,t}$, $i=1,2,3$ and $j=1,2$, are independent realizations of $\varepsilon^{i}_{t}$.

\paragraph{Design of the simulation study.}
We consider $\gamma$ and $\eta$ each distributed according to one of the five GP models from \eqref{eq:PowerTestProcesses} from which we simulate samples $\gamma_1,...,\gamma_N$ and $\eta_1,...,\eta_N$, $N\in \{10,15,30\}$, on  $[0,1]$ with an equidistant time grid of width $\Delta t= 0.01$. We call these \emph{aligned samples} and we subject the resulting $N(N+1)/2$ pairings to Test \ref{test:ILLP}, cf. Table \ref{table:MILLP} (a). Further for fixed rotations $P\in SO(3)$ with Euler angles $\alpha_x=-0.5^\circ$, $\alpha_y=13^\circ$, $\alpha_z=-9^\circ$ and $Q\in SO(3)$ with Euler angles $\alpha_x=12^\circ$, $\alpha_y=0^\circ$, $\alpha_z=5^\circ$ we misalign each sample $\eta_1,...,\eta_N$ to obtain $\tilde\eta_1,...,\tilde\eta_N$ via $\tilde\eta_n = P \eta_n Q$ ($n=1,\ldots,N$). We call these the \emph{misaligned samples} and we subject the resulting $N(N+1)/2$ pairings to Tests \ref{test:ILLP}, \ref{test:PILLP} and \ref{test:MILLP}, cf. Table \ref{table:MILLP} (b), (c) and (d).
% 
% and \ref{test:MILLP}
% 
% 
% If we do not include the action of an isometry we apply the ILLPerm and MILLPerm to these samples. In order to explore the robustness of the proposed testing procedures, if one has to correct for the perturbation by an element from $\Inull$, we compute the realization $P\eta_1Q^T,...,P\eta_NQ^T$ with a rotation $P\in SO(3)$ with Euler angles $\alpha_x=-0.5^\circ$, $\alpha_y=13^\circ$, $\alpha_z=-9^\circ$ and a rotations $Q\in SO(3)$ with Euler angles $\alpha_x=12^\circ$, $\alpha_y=0^\circ$, $\alpha_z=5^\circ$. Denote this rotated sample by $\tilde\eta_1,...,\tilde\eta_N$. For ILLPerm we estimate $\hat P$ and $\hat Q$ as described in Section \ref{} and apply it to the samples $\gamma_1,...,\gamma_N$ and $\hat P^TP\eta_1Q^T\hat Q,...,\hat P^TP\eta_NQ^T\hat Q$. MILLPerm can be applied directly to the two samples $\gamma_1,...,\gamma_N$ and $\tilde\eta_1,...,\tilde\eta_N$.

We use $M=2000$ simulations and for each  $M_{perm} = 5000$ permutations. We do not simulate time warping actions because, in contrast to spatial alignment which can be explicitly estimated via a SVD, running for every permutation for every simulation the dynamic program from Section \ref{scn:generic-loss-fcn} would require in our implementation 70.000 hours on an Intel(R) Core(TM) i5-4300U CPU with 1.90GHz.

\paragraph{Results of the simulation study.}
In Table \ref{table:MILLP} we report the corresponding acceptance rates. Subtracting the diagonal elements from one gives the type I errors, the off-diagonals give type II errors.  For aligned samples, Test \ref{test:ILLP} keeps the level (Table \ref{table:MILLP} (a)) and so does Test \ref{test:MILLP} for misaligned data (Table \ref{table:MILLP} (d)). For misaligned data, Test \ref{test:ILLP} completely fails to identify  equality up to the spacial action (Table \ref{table:MILLP} (b)) while Test  \ref{test:PILLP} (Table \ref{table:MILLP} (c)) masters identification but at a  lower level, at the price of a lower power. We have also applied  Test \ref{test:MILLP} to aligned data which gives a table similar to  Table \ref{table:MILLP} (d).

\begin{table}[!ht]\tiny
    \begin{subtable}[c]{0.5\textwidth}\centering
    \begin{tabular}{| c | r | r | r | r | r |}
		      \hline
		      \multicolumn{1}{|r|}{\tiny \textbf{10}}	&  				& 				&  				&             			& \\
		      \tiny \textbf{N=15}			& $\gamma^0_{A}$ 		& $\gamma^{0.5}_{B}$ 		& $\gamma^{1}_{B}$ 		& $\gamma^{2}_{B}$            	& $\gamma^{2.5}_{B}$\\
		      \multicolumn{1}{|r|}{\tiny \textbf{30}}	&  				& 				&  				&             			& \\		  
		      \hline\hline                       
		      $\gamma^0_{A}$  		&  \makecell{95.8\\94.5\\94.5}	& \makecell{80.7\\72.8\\45.1}	&  \makecell{29.8\\8.6\\0.0}	& \makecell{0.0\\0.0\\0.0}	& \makecell{0.0\\0.0\\0.0} \\\hline
		      $\gamma^{0.5}_{B}$  	&  \cellcolor{gray}		& \makecell{95.4\\94.6\\94.8}	&  \makecell{88.1\\81.6\\63.0}	& \makecell{8.4\\0.2\\0.0}	& \makecell{0.1\\0.0\\0.0}\\\hline
		      $\gamma^{1}_{B}$  	&  \cellcolor{gray}		& \cellcolor{gray}		&  \makecell{94.2\\95.3\\93.65}	& \makecell{49.7\\26.4\\0.3}	& \makecell{7.1\\0.3\\0.0}\\\hline
		      $\gamma^{2}_{B}$  	&  \cellcolor{gray}		& \cellcolor{gray}		&  \cellcolor{gray}		& \makecell{95.3\\95.3\\95.4}	& \makecell{87.9\\81.2\\61.8}\\\hline
		      $\gamma^{2.5}_{B}$  	&  \cellcolor{gray}		& \cellcolor{gray}		&  \cellcolor{gray}		& \cellcolor{gray}		& \makecell{95.3\\95.2\\95.4}\\
		      \hline  
    \end{tabular}
    \caption{Test \ref{test:ILLP} on aligned samples.}
    \end{subtable}
    \begin{subtable}[c]{0.5\textwidth}\centering
    \begin{tabular}{| c | r | r | r | r | r |}
		      \hline
		      \multicolumn{1}{|r|}{\tiny \textbf{10}}	&  				& 				&  				&             			& \\
		      \tiny \textbf{N=15}			& $\gamma^0_{A}$ 		& $\gamma^{0.5}_{B}$ 		& $\gamma^{1}_{B}$ 		& $\gamma^{2}_{B}$            	& $\gamma^{2.5}_{B}$\\
		      \multicolumn{1}{|r|}{\tiny \textbf{30}}	&  				& 				&  				&             			& \\		  
		      \hline\hline                       
		      $\gamma^0_{A}$  		&  \makecell{0.0\\0.0\\0.0}	& \makecell{0.0\\0.0\\0.0}	&  \makecell{0.0\\0.0\\0.0}	& \makecell{0.0\\0.0\\0.0}	& \makecell{0.0\\0.0\\0.0} \\\hline
		      $\gamma^{0.5}_{B}$  	&  \cellcolor{gray}		& \makecell{0.0\\0.0\\0.0}	&  \makecell{0.0\\0.0\\0.0}	& \makecell{0.0\\0.0\\0.0}	& \makecell{0.0\\0.0\\0.0}\\\hline
		      $\gamma^{1}_{B}$  	&  \cellcolor{gray}		& \cellcolor{gray}		&  \makecell{0.0\\0.0\\0.0}	& \makecell{0.0\\0.0\\0.0}	& \makecell{0.0\\0.0\\0.0}\\\hline
		      $\gamma^{2}_{B}$  	&  \cellcolor{gray}		& \cellcolor{gray}		&  \cellcolor{gray}		& \makecell{0.0\\0.0\\0.0}	& \makecell{0.0\\0.0\\0.0}\\\hline
		      $\gamma^{2.5}_{B}$  	&  \cellcolor{gray}		& \cellcolor{gray}		&  \cellcolor{gray}		& \cellcolor{gray}		& \makecell{0.0\\0.0\\0.0}\\
		      \hline  
    \end{tabular}
    \caption{Test \ref{test:ILLP} on misaligned samples.}
    \end{subtable}
    
    \vspace*{0.5cm}
%   \caption{Acceptance rate in percentage of $H_0$ of ILLPerm (left) and MILLPerm (right).}
%    \label{table:MILLP}
% \end{table}
% 
% \begin{table}[!htp]\tiny
    \begin{subtable}[c]{0.5\textwidth}\centering
    \begin{tabular}{| c | r | r | r | r | r |}
		      \hline
		      \multicolumn{1}{|r|}{\tiny \textbf{10}}	&  				& 				&  				&             			& \\
		      \tiny \textbf{N=15}			& $\gamma^0_{A}$ 		& $\gamma^{0.5}_{B}$ 		& $\gamma^{1}_{B}$ 		& $\gamma^{2}_{B}$            	& $\gamma^{2.5}_{B}$\\
		      \multicolumn{1}{|r|}{\tiny \textbf{30}}	&  				& 				&  				&             			& \\		  
		      \hline\hline                       
		      $\gamma^0_{A}$  		&  \makecell{99.9\\99.9\\99.9}	& \makecell{93.6\\90.2\\73.1}	&  \makecell{50.1\\16.5\\0.0}	& \makecell{0.0\\0.0\\0.0}	& \makecell{0.0\\0.0\\0.0} \\\hline
		      $\gamma^{0.5}_{B}$  	&  \cellcolor{gray}		& \makecell{96.1\\95.9\\96.5}	&  \makecell{93.1\\88.4\\75.9}	& \makecell{11.5\\0.8\\0.0}	& \makecell{0.2\\0.0\\0.0}\\\hline
		      $\gamma^{1}_{B}$  	&  \cellcolor{gray}		& \cellcolor{gray}		&  \makecell{95.8\\96.8\\96.1}	& \makecell{63.4\\39.0\\0.8}	& \makecell{12.4\\0.3\\0.0}\\\hline
		      $\gamma^{2}_{B}$  	&  \cellcolor{gray}		& \cellcolor{gray}		&  \cellcolor{gray}		& \makecell{96.4\\96.6\\96.3}	& \makecell{91.7\\89.1\\75.1}\\\hline
		      $\gamma^{2.5}_{B}$  	&  \cellcolor{gray}		& \cellcolor{gray}		&  \cellcolor{gray}		& \cellcolor{gray}		& \makecell{96.7\\96.4\\97.1}\\
		      \hline  
    \end{tabular}
    \caption{Test \ref{test:PILLP} on misaligned samples.}
    \end{subtable}
    \begin{subtable}[c]{0.5\textwidth}\centering
    \begin{tabular}{| c | r | r | r | r | r |}
		      \hline
		      \multicolumn{1}{|r|}{\tiny \textbf{10}}	&  				& 				&  				&             			& \\
		      \tiny \textbf{N=15}			& $\gamma^0_{A}$ 		& $\gamma^{0.5}_{B}$ 		& $\gamma^{1}_{B}$ 		& $\gamma^{2}_{B}$            	& $\gamma^{2.5}_{B}$\\
		      \multicolumn{1}{|r|}{\tiny \textbf{30}}	&  				& 				&  				&             			& \\		  
		      \hline\hline                       
		      $\gamma^0_{A}$  		&  \makecell{94.9\\94.9\\95.4}	& \makecell{81.8\\74.0\\40.0}	&  \makecell{24.1\\2.8\\0.0}	& \makecell{0.0\\0.0\\0.0}	& \makecell{0.0\\0.0\\0.0} \\\hline
		      $\gamma^{0.5}_{B}$  	&  \cellcolor{gray}		& \makecell{95.3\\95.2\\94.3}	&  \makecell{89.3\\86.1\\70.5}	& \makecell{9.4\\0.2\\0.0}	& \makecell{0.0\\0.0\\0.0}\\\hline
		      $\gamma^{1}_{B}$  	&  \cellcolor{gray}		& \cellcolor{gray}		&  \makecell{95.1\\95.3\\95.8}	& \makecell{59.7\\30.2\\0.4}	& \makecell{8.0\\0.2\\0.0}\\\hline
		      $\gamma^{2}_{B}$  	&  \cellcolor{gray}		& \cellcolor{gray}		&  \cellcolor{gray}		& \makecell{95.2\\94.9\\95.0}	& \makecell{89.1\\85.4\\67.8}\\\hline
		      $\gamma^{2.5}_{B}$  	&  \cellcolor{gray}		& \cellcolor{gray}		&  \cellcolor{gray}		& \cellcolor{gray}		& \makecell{95.7\\94.8\\95.0}\\
		      \hline  
    \end{tabular}
        \caption{Test \ref{test:MILLP} on misaligned samples.}
    \end{subtable}
  \caption{\it Average acceptances rate in percent over $M=2000$ simulations, testing  (\ref{eq:NullHypoRGP}).}
%   \label{table:MILLPalign}
   \label{table:MILLP}
\end{table}

% % % % % % % % % % % % % % % % % % % % % % % % % % % Applications to Gait % % % % % % % % % % % % % % % % % % % % % % % % % % %
% \section{The Marker Replacement Effect in Gait Analysis}\label{scn:Application}
\section{Applications To Gait Data}\label{scn:Application}

In this application we study biomechanical gait data collected for an experiment performed in the School of Rehabilitation Science (McMaster University, Canada). 
In order to identify individual gait patterns in a common practice clinical setting, subjects walked along a 10 m pre-defined straight path at comfortable speed. Retro-reflective markers were placed by an experienced technician onto identifiable skin locations on the upper and lower leg, following a standard protocol. 
Eight cameras recorded the position of these markers and from these spatial marker motions, a moving orthogonal frame $E_u(t)\in SO(3)$ describing the rotation of the upper leg w.r.t. the laboratory's fixed coordinate system was determined, and one for the lower leg, $E_l(t) \in SO(3)$, each of which was aligned near $I_{3\times 3}$ when the subject stood straight. 

Eight subjects  repeatedly walked along the straight path and for each walk a single gait cycle $\gamma(t) = E_u(t) E_l(t)^T$ representing the motion of the upper leg w.r.t. the lower leg was recorded. Due to Newton's law, these curves are continuously differentiable, i.e., $\gamma \in \X$. By design each subject started walking about three gait cycles prior to data collection allowing the assumption of independence of recorded gait cycles. Two samples ("before" and "after" marker replacement) of  $N\approx 12$ walks were available for analysis. 

Replacing markers between sessions results then in fixed and different rotations of the upper and lower legs conveyed by suitable $P,Q \in SO(3)$ such that
$E'_u(t) = P E_u(t)$ and $E'_l(t) = Q^T E_l(t)$. In consequence of Lemma \ref{lem:isometry-group}, marker replacement corresponds to the action of $\Inull$ on gait curves in $\X$, 
\begin{align*} \gamma'(t) &= E'_u(t) E'_l(t)^T = P \gamma(t)Q\,.\end{align*}
Although, $P$ and $Q$ are near $I_{3\times 3}$, it has been known (see \cite{noehren2010improving} and \cite{groen2012sensitivity}) that this \emph{marker replacement effect} makes identification of individual gait patterns challenging.
This is illustrated in Figure \ref{fig:marker-replacement-effect}a: While the effect is small on the dominating Euler flexion-extension angle, it can amount to a considerable distortion of other angles; here the Euler internal-external angle is mainly shifted.

Indeed, applying Test \ref{test:ILLP} to the gait samples before and after marker replacement shows that 6 out of 8 subjects cannot be identified with significance level $\alpha =5\,\%$, cf. Table \ref{table:SSAnecessary}a. If, however, the group action due to $\Inull$, which precisely models the marker replacement effect, is taken into account in addition with sample-wise time warping, by Test \ref{test:MILLP}, then all 8 subjects can be reliably identified and well discriminated from other subjects, cf. Table \ref{table:SSAnecessary}b. This preliminary study warrants larger investigations to develop a protocol allowing for inferential gait analysis robustly modeling marker replacement and temporal warping effects.

\begin{table}[t!]\scriptsize\centering
\begin{subtable}{.45\textwidth}
\setlength{\tabcolsep}{1mm}
\caption{\textit{Test 4.1 without any spatial and \newline temporal alignment.}}
% 	      \multicolumn{9}{c}{\textbf{left knee}} \\
	    \begin{tabular}{|c|cccccccc|}
	      \hline
		\textbf{Vol}	& $1$ & $2$ & $3$ & $4$ & $5$ & $6$ & $7$ & $7$  \\ 
	      \hline
	      $1$ 	& {$0.6$} & 0 & 0 & 0 & 0 & 0 & 0 & 0 \\ 
	      $2$  	& 0 & {$0.0$} & 0 & 0 & 0 & 0 & 0 & 0 \\ 
	      $3$ 	& 0 & 0 & {$2.7$} & 0 & 0 & 0 & 0 & 0 \\  
	      $4$ 	& 0 & 0 & 0 & {$0.1$} & 0 & 0 & 0 & 0 \\
	      $5$ 	& 0 & 0 & 0 & 0 & { $0.0$} & 0 & 0 & 0 \\ 
	      $6$ 	& 0 & 0 & 0 & 0 & 0 & $6.5$ & 0 & 0 \\
	      $7$ 	& 0 & 0 & 0 & 0 & 0 & 0 & $6.8$ & 0 \\ 
	      $8$ 	& 0 & 0 & 0 & 0 & 0 & 0 & 0 & {$1.6$} \\
	      \hline
	     \end{tabular}
\end{subtable}%
\begin{subtable}{.45\textwidth}
\setlength{\tabcolsep}{1mm}
\caption{{\it Test 4.4 including spatial and \newline temporal alignment.}}
	    \begin{tabular}{|c|cccccccc|}
% 	      \multicolumn{9}{c}{\textbf{left knee}} \\
	      \hline
		\textbf{Vol}	& $1$ & $2$ & $3$ & $4$ & $5$ & $6$ & $7$ & $7$  \\ 
	      \hline
	      $1$ 	& $23.6$ & 0 & 0 & 0 & 0 & 0 & 0 & 0 \\ 
	      $2$  	& 0 & $39.2$ & 0 & 0 & 0 & 0 & 0 & 0 \\ 
	      $3$ 	& 0 & 0 & $66.3$ & 0 & 0 & 0 & 0 & 0 \\  
	      $4$ 	& 0 & 0 & 0 & $53.8$ & 0 & 0 & 0 & 0 \\
	      $5$ 	& 0 & 0 & 0 & 0 & $6.2$ & 0 & 0 & 0 \\ 
	      $6$ 	& 0 & 0 & 0 & 0 & 0 & $53.7$ & 0 & 0 \\
	      $7$ 	& 0 & 0 & 0 & 0 & 0 & 0 & $97.2$ & 0 \\ 
	      $8$ 	& 0 & 0 & 0 & 0 & 0 & 0 & 0 & $60.7$ \\
	      \hline
	     \end{tabular}
\end{subtable}
\caption{\textit{Displaying $p$-values in percentage of permutation tests of equality of means before and after marker replacement for different volunteers.}}
\label{table:SSAnecessary}
\end{table}

\section*{Acknowledgment}

The first and the last author gratefully acknowledge DFG HU 1575/4 and the Niedersachsen Vorab of the Volkswagen Foundation.

% 
% 
% {\color{red}------------------ That belongs into the ILL section ------------------
%   \begin{Rm}\label{remark:ILL}
% 	    In view of our biomechanical application, part $(iii)$ of the above theorem shows that $\delta_{I,1}$ and $\delta_{I,2}$ are invariant under any changes due to marker placement of one of the two orthogonal coordinate frames constructed from the measurement device.
%    \end{Rm}
% ------------------ That belongs into the ILL section ------------------}
% % % % % % % % % % % % % % % % % % % % % % % % % % % Discussion % % % % % % % % % % % % % % % % % % % % % % % % % % %
% \section{Discussion}

\FloatBarrier
\appendix
% \section{Appendix: Marker Placement Effect}\label{appendix:MPeffect}
% {\bf ----------------- end not touched yet -----------------}

\section{Appendix: Proofs}

For convenience, here we will always use $I=[0,1]$.

\paragraph{Proof of Lemma \ref{lem:isometry-group}.}
    Using \cite[Theorem 4.1(i) on p. 207]{Helgason1962}, we have to assert that $\cG := \G\times \G$ and $\cK:=\diag(\G\times \G)$ form a \emph{Riemannian symmetric pair} where  $\cG$ is semisimple  and acts effectively on $\cG/\cK = \{[g,h]: g,h \in \G\}$, $[g,h] = \{(gk,hk): k\in \G\}$. Here the action is given by  $(g',h'):[g,h] \mapsto  [g'g,h'h]$. The fact that $(\cG,\cK)$ is a Riemannian symmetric pair is asserted in \cite[p. 207]{Helgason1962}, $\G=SO(3)$ is simple, hence semisimple and the effective action follows from the fact $\cG$ has no trivial normal divisors $N \subset \cK$, cf.  \cite[p. 110]{Helgason1962}. For if $\{(e,e)\}\neq N\subset \cK$ would be a normal divisor of $\cG$ then there would be a subgroup $\{e\}\neq H$ of $\G$ with the property that for every $h\in H, g,k\in \G$, in particular for $g\neq k$, there would be $h'\in H$ such that
    $$ (g,k) (h,h) (g^{-1},k^{-1}) = (h',h')\,,$$
    i.e. $k^{-1}g$ would be in the center of $\G$, which, however, is trivial because $\G$ is non- commutative and simple. Hence $g=k$, a contradiction.
\qed

\paragraph{Proof of Theorem \ref{thm:FullAsymptoticEquivalence}.}
% 	Recall that the notation $\max_{t\in I}\Vert C_{t} \Vert =  \mathcal{O}_p(\sigma)$ as $\sigma\rightarrow0$ means that for all sequences $\N\ni l\mapsto\sigma_l>0$ with $\sigma_l\rightarrow0$ as $l\rightarrow\infty$ and  a sequence of processes $C^l_{t}$ we have that for for all $\varepsilon>0$ there exists a $M>0$ such that
% 	\begin{equation*}
% 	 \Prb\Big( \sigma_l^{-2} \max_{t\in I}\big\Vert C^l_{t} \big\Vert > M \Big) < \varepsilon
% 	\end{equation*}
% 	for all $l\in\N$. Therefore in the following we assume that $\sigma_l\rightarrow0$ and consider a sequence of Gaussian perturbation models with Gaussian processes $\big\{C^l_t\big\}_{t\in I}$ and $\big\{D^l_t\big\}_{t\in I}$ such that $\max_{t\in I}\big\Vert C^l_{t} \big\Vert = \mathcal{O}_p(\sigma_l) = \max_{t\in I}\big\Vert D^l_{t} \big\Vert$.
	In the following we set  $\|C_t\|_\infty = \max_{t\in I}\big\Vert C_t\big\Vert$, for short.
	
	First observe that (\ref{iota-rotation:eq}) and (\ref{exp-commutator:eq}) imply
	\begin{equation*}
	    \delta_0(t)^T\Exp\Big(\iota \circ C_t\Big)\, \delta_0(t)\,\Exp\Big(\iota \circ D_t\Big)
			  = \Exp\Bigg(\iota \circ \!\left(\delta_0(t)^TC_t\right)\Bigg)\,\Exp\Big(\iota \circ D_t\Big)\,.
	\end{equation*}
	Thus, since $\big\Vert\delta_0(t)^TC_t\big\Vert_\infty = \big\Vert C_t\big\Vert_\infty = \mathcal{O}_p(\sigma)$, we may w.l.o.g. set $\delta_0(t)=I_{3\times3}$ for all $t\in I$.
       
% 	Further, 
% 	by $\max_{t\in I}\big\Vert C_t\big\Vert =  \mathcal{O}_p(\sigma)$ we have that for every $\epsilon>0$ there is $M_\epsilon>0$ such that
% 	\begin{equation*}
% 	\Prb\!\left(\sigma^{-1}\max_{t\in I}\big\Vert C_t \big\Vert > M_\epsilon \right) < \epsilon
% 	\end{equation*}
% 	for all $\sigma >0$. In consequence,  for every $\delta>0$ there is $\sigma_\delta>0$ such that for all $\sigma < \sigma_\delta$, we have that $\sigma^{-1}\delta > M-\epsilon$, which implies
% 	\begin{equation*}
% 	\Prb\!\left( \max_{t\in I}\big\Vert C_t \big\Vert > \delta \right) = \Prb\!\left(\sigma^{-1}\max_{t\in I}\big\Vert C_t \big\Vert > \sigma^{-1}\delta \right) \leq \Prb\!\left(\sigma^{-1}\max_{t\in I}\big\Vert C_t \big\Vert > M \right) <\epsilon
% 	\end{equation*}
% 	for all  $\sigma < \sigma_\delta$. This means that
% 
% 	
% 	we have $C^l,D^l\xrightarrow \Prb 0$ for $n\rightarrow\infty$ with respect to the maximums norm. This can be seen as follows:
% 
% 	 Thus, since $\sigma\rightarrow 0$,
% 
% 	Now, 
	Next, define the group and vector valued valued processes
	\begin{equation*}
	      h_t = \Exp\!\left(\iota\circ C_t\right)\Exp\!\left(\iota\circ D_t\right)\,,
% 	\end{equation*}
	\mbox{ and }
% 	\begin{equation*}
	  H_t = \iota^{-1}\circ \Exp\vert_\mathfrak{V}^{-1} h_t\,,
	\end{equation*}
	where $\mathcal{B}_\pi(0)\subset \mathfrak{V} \subset\R^3$ is a set making the Lie exponential restricted to $\mathfrak{V}$ bijective onto $G$.  In consequence, setting $A_t = C_t + D_t$ and 
	\begin{equation*}
	      \tilde A_t =  H_t - A_t\,,
	\end{equation*}
	we have 
	\begin{equation*}
	      \Exp\!\left(\iota\circ A_t + \iota\circ\tilde A_t\right) =  \Exp\!\left(\iota \circ C_t\right)\,\Exp\!\left(\iota \circ D_t\right)\,,
	\end{equation*}
	and in order to prove the theorem, it suffices to show that for all $\epsilon > 0$ there is $M=M_\epsilon>0$ such that
	\begin{equation}\label{eq:goal}
	    \Prb\!\left\{ \sigma^{-2} \big\Vert \tilde A_t\big\Vert_\infty > M\right\} < \epsilon\,.
	\end{equation}

	To this end, fix $\epsilon > 0$ and consider the function
	\begin{equation*}
	  f:\so(3)\times \so(3)\to \so(3),~(X,Y) \mapsto\Log\big(\Exp(X)\,\Exp(Y)\big)
	\end{equation*}
	which is well-defined and analytic in a neighborhood $U$ of $(0,0) \in \so(3)\times \so(3)$. In particular, there is a compact subset $K\subset U$ which contains a neighborhood of $(0,0)$ such that the Taylor expansion
	\begin{equation}\label{eq:3proof-thm-equivalence-FOGPM} 
	  f(X,Y) = X + Y + \mathcal{O}\big(\Vert X \Vert^2 +\Vert Y \Vert^2\big)
	\end{equation}
	is valid for all $(X,Y)\in K$.

	We now split the l.h.s. of \eqref{eq:goal} into the two summands
	\begin{equation*}
	\Prb\!\left\{ \mathfrak{K};~ \sigma^{-2}\big\Vert  \iota^{-1}\circ\Log\big(h_t\big) - C_t - D_t \big\Vert_\infty > M \right\}
% 	\end{equation*}
	\mbox{ and }
% 	\begin{equation*}
	\Prb\!\left\{ \mathfrak{K}^C;~ \sigma^{-2} \big\Vert  H_t - C_t - D_t \big\Vert_\infty > M \right\}\,,
	\end{equation*}
	where $\mathfrak{K}= \{ \omega\in \Omega\,\vert~\forall t \in I:~ \big(\iota\circ C_t,\iota\circ D_t\big) \in K\}$. Note that we used in the first summand that $\iota \circ\Log\big(h_t\big)=H_t$ for all $t\in I$, if the complete sample path $t\mapsto(\iota\circ C_t,\iota\circ D_t)$ is contained in $K$.

	For the first summand we have $\omega \in \mathfrak{K}$ and Taylor's formula \eqref{eq:3proof-thm-equivalence-FOGPM} is applicable, yielding
	\begin{align*}
	  \big\Vert \tilde A_t \big\Vert_\infty \leq \Lambda \max_{t\in I}\!\left(\big\Vert C_t \big\Vert^2 +\big\Vert D_t \big\Vert^2\right)
				  \leq \Lambda \Vert C_t \Vert_\infty^2 + \Lambda \Vert D_t \Vert_\infty^2\,, %= \mathcal{O}_p\big(\sigma^2\big)\,,
	\end{align*}
	where $\Lambda>0$ can be chosen independent of $t\in I$, since the Hessian of $f$ is bounded on $K$.
	Thus, $\max_{t\in I}\big\Vert \tilde A_t \big\Vert = \mathcal{O}_p\big(\sigma^2\big)$ on $\mathfrak{K}$ and therefore we find $M>0$ such that 
	\begin{equation}\label{eq:4proof-thm-equivalence-FOGPM}
	\Prb\!\left\{ \mathfrak{K};~ \sigma^{-2}  \big\Vert \iota^{-1}\circ \Log\big(h_t\big) - C_t - D_t \big\Vert_\infty \geq M \right\} < \frac{\epsilon}{2}\,,
	\end{equation}
	
	The second summand is characterized by the fact that there is  $t\in I$ with $\big(\iota\circ C_t, \iota\circ D_t\big)\notin K$. Since $\|C_t\|_\infty, \|D_t\|_\infty\xrightarrow \Prb 0$ , it follows that
	$\Prb\!\left( \mathfrak{K}^C  \right) \rightarrow 0$ as $\sigma \to 0$. Hence there is $\sigma_0>0$ such that for all $\sigma<\sigma_0$,
	\begin{equation*}
	\Prb\Big( \mathfrak{K}^C \Big) < \frac{\epsilon}{2}\,.
	\end{equation*}
	In consequence, for any $M>0$ and all $\sigma < \sigma_0$,
	\begin{align*}
	\Prb\!\left\{ \mathfrak{K}^C;~ \sigma^{-2}  \big\Vert H_t - C_t - D_t \big\Vert_\infty > M \right\} \leq
% 	\\
	\Prb\!\left\{ \big(\iota\circ C_t, \iota\circ D_t\big) \in K^C \mbox{ for some }t\in I\right\} < \frac{\epsilon}{2}\,.
	\end{align*}
	
	In conjunction with \eqref{eq:4proof-thm-equivalence-FOGPM}, this yields \eqref{eq:goal} completing the proof.
\qed

\paragraph{Proof of Theorem \ref{theorem:ComputingPEM}.}
(i) follows at once from the arguments preceding the theorem. To see (ii) we note that the sum of $\cC^1$ functions is again $\cC^1$ and the derivative of the expected value is the expected value of derivatives in case of dominated convergence of the derivative. 
Now consider
   $\Gamma = \{\E[\gamma(t)]: t\in [0,1]\}$, $\Gamma_N = \{\bar\gamma_N(t)]: t\in [0,1]\}$ and ${\mathcal F} =\big\{B\in \R^{3\times 3}: \rank(B)\leq 1\big\}$.
%   \begin{align*}
% %   \label{pointset-mean:eq}
%   \left.\begin{array}{rcll}
%    \Gamma &= &\{\E[\gamma(t)]: t\in [0,1]\},& \Gamma_N = \{\bar\gamma_N(t)]: t\in [0,1]\},\mbox{ and } \\{\mathcal F} &=& \big\{B\in \R^{3\times 3}: \rank(B)\leq 1\big\}\end{array}\right\}\,.
%   \end{align*}
% which gives at once
%   \begin{align}\label{unique:eq}
% %   \left.\begin{array}{rl}
%    \mbox{ the population (sample) PEM is unique }&\Leftrightarrow \Gamma\cap {\mathcal F}=\emptyset ~~\big(\Gamma_N\cap {\mathcal F}=\emptyset\big) \,.
% %    \\
% %    \mbox{ the sample PEM is unique }&\Leftrightarrow \Gamma_N\cap {\mathcal F}=\emptyset\end{array}\right\}.
%   \end{align}
% In that case, 
Because $\Gamma$ and $\Gamma_N$ are compact and ${\mathcal F}$ is closed, in case of uniqueness there  are open sets $U \supset \Gamma$, $U_N\supset \Gamma_N$ such that $U\cap {\mathcal F} = \emptyset=U_N\cap {\mathcal F}$. By \citet[Theorem  4.1, p. 6]{Dudek1994}, $\mathfrak{pr}|_U$ and $\mathfrak{pr}|_{U_N}$ are even analytic, yielding that the composition is $\cC^1$ as claimed.
\qed

% \paragraph{Proof of Theorem \ref{theorem:PEMisCenterCurve}.}
\begin{Lem}\label{lemma:ExpSymmetry} Suppose that the distribution of a random variable $X\in \R^3$ with existing first moment $\E[X]$ is even, i.e., that $\Prb\{X\in M\} = \Prb\{-X\in M\}$ for all Borel sets $M\subset \R^3$ and $\E\left[ \cos\Vert X \Vert\right]>0$. Then
      $ \E\big[\Exp(\iota \circ X)\big]$ is symmetric positive definite.
\end{Lem}
\begin{proof} 
      With the Rodriguez formula \eqref{eq:rodriguez} we have
      \begin{align*}
	\Exp(\iota\circ X) &= I_{3\times 3} + \iota\circ X\, \sinc{ \Vert X \Vert } +( \iota\circ X)^2\,\frac{1-\cos \Vert X\Vert}{\Vert X\Vert ^2} \\
		           &= \cos(\Vert X\Vert)I_{3\times 3} +\iota\circ X\,  \sinc{ \Vert X \Vert }  + (1-\cos(\Vert X\Vert)) \frac{XX^T}{\Vert X \Vert^2}\,,
      \end{align*}
       where the second equality is due to
    \begin{equation*}%\label{eq:iota-product}
    \left( \iota\left(\frac{X}{\Vert X \Vert}\right) \right)^2 = \frac{XX^T}{\Vert X \Vert^2} - I_{3\times 3}\,.
    \end{equation*}
     By hypothesis, $X$ is even, hence $\E[\iota\circ X\, \sinc{ \Vert X \Vert }]=\iota \circ \E[ X\, \sinc{ \Vert X \Vert }]=0$ which yields that
     \begin{equation*} 
      \E[\Exp(\iota\circ X)] = \E[\cos\Vert X\Vert] I_{3\times 3} + \E\left[(1-\cos\Vert X\Vert) \frac{XX^T}{\Vert X \Vert^2}\right]
     \end{equation*}
      is symmetric. Let $V \in \R^3$ be arbitrary with $\Vert V\Vert =1$, then positive definiteness follows from the inequality
      \begin{equation*}
	V^T \E\left[ \Exp(\iota\circ X) \right] V = \E[\cos\Vert X\Vert ] + \E\left[(1-\cos\Vert X\Vert ) \frac{(V^TX)^2}{\Vert X\Vert ^2}\right] \geq \E[\cos\Vert X\Vert ]>0\,.
      \end{equation*}
%       where the last inequality is due to the Cauchy-Schwartz inequality and $1-\cos\Vert X\Vert >0$.
\end{proof}
\paragraph{Proof of Theorem \ref{theorem:PEMisCenterCurve}.}
    Fixing $t\in I$ and letting $\E[\gamma(t)] = UDV^T$ be a SVD, with Theorem \ref{theorem:ComputingPEM}, we need to assert $\gamma_0(t) = U S V^T$.
        Since $\E[\gamma(t)] = \gamma_0(t) \E\!\left[\Exp\big(\iota \circ A_t\big)\right]$, this holds if $\E\!\left[\Exp\big(\iota\circ A_t\big)\right]$ is symmetric and positive definite.
    
    Indeed, let $\E\!\left[\Exp\big(\iota\circ A_t\big)\right]=\tilde V \Lambda \tilde V^T$ with $\tilde V \in  SO(3)$ and diagonal $\Lambda>0$. Then $\Lambda=D$ (assuming that the singular values are sorted non increasingly), %since $\E[\gamma(t)]^T\E[\gamma(t)]=\E\!\left[\Exp\big(\iota\circ A_t\big)\right]^T\!\E\!\left[\Exp\big(\iota\circ A_t\big)\right]$, 
    giving 
    \begin{equation*}
    UDV^T=\gamma_0(t) \tilde V D \tilde V^T\,.
    \end{equation*}
    Since the eigenspaces corresponding to same eigenvalues spanned by the columns of $\tilde V$ and $V$ agree, we can choose $\tilde V=V$, yielding $\gamma_0 = UV^T$ with $U\in SO(3)$, since $U=\gamma_0(t) \tilde V\in SO(3)$, which proves the assertion, even with $S= I_{3\times 3}$.

    With Lemma \ref{lemma:ExpSymmetry} it remains to prove that $A_t=X \sim\Norm(0,\Sigma)$ in $\R^3$ fulfills $\E[\cos\Vert X\Vert ] > 0$. Indeed, making use of the Fourier transformation of Gaussian densities in the 3rd equality below, we have
    \begin{align*}
    \E[\cos\Vert X\Vert ] % &= \frac{1}{(2\pi)^{k/2}\sqrt{\nu_1\cdot...\cdot\nu_k}} \int_{\R^k} e^{-x^T\,\widetilde{\nu}^{-1}x/2} \cos(\Vert x \Vert )\,dx \\
			    &= \frac{1}{(2\pi)^{k/2}\sqrt{\nu_1\cdot...\cdot\nu_k}} \int_{\R^k} e^{-y^T\,\widetilde{\nu}^{-1}y/2} \cos(\Vert y\Vert )\,dy\\
			    &= \frac{1}{(2\pi)^{k/2}}\int_{S^{k-1}}\left(\int_0^\infty e^{-r^2/2} \cos \left(r \sqrt{\phi^T\widetilde{\nu} \phi}\right)\,dr\right)\,d\sigma(\phi)\\
% 			   &= \frac{1}{(2\pi)^{k/2}}\int_{S^{k-1}} \left( \frac{1}{2}\,\int_{-\infty}^\infty e^{-r^2/2} e^{ i r \sqrt{\phi^T\widetilde{\nu} \phi} } \,dr \right)\,d\sigma(\phi)\\
			    &= \frac{1}{(2\pi)^{k/2-1}}\int_{S^{k-1}} \frac{1}{2}\,e^{-\phi^T\,\widetilde{\nu}\phi/2}\, d\sigma(\phi)~>~0\,,
    \end{align*}
    with the spherical volume element $d\sigma(\phi)$ on the $k-1$ dimensional unit sphere $S^{k-1}$. Here we have used a svd $\Sigma = W \,\diag(\nu_1,\nu_2,\nu_3)W^T$ with $W=(w_1,w_2,w_3)\in SO(3)$, the smallest index $k\in \{1,2,3\}$ such that $\nu_k>0$, $\widetilde{\nu} =\diag(\nu_1,...,\nu_k)$  and $y=(w_1,...,w_k)^Tx \in \R^k$.
\qed

\paragraph{Proof of Corollary \ref{corollary:PESMconsistency}.}
  That the extrinsic sample mean set is a strongly consistent estimator of the extrinsic population mean set follows from a more general result by \citet{Ziezold77}. In case of uniqueness, guaranteed by the model we have that the center curve and PEM agree by virtue of Theorem \ref{theorem:PEMisCenterCurve}. Hence, every measurable selection of the sample mean converges almost surely to the unique population mean  yielding the assertion (see also \citet{Bhattacharya2003}).
\qed

\paragraph{Proof of Theorem \ref{theorem:PESMuniformconvergence}.} First, we claim that in order to prove (i) it suffices to show that
    \begin{equation}\label{equation:uniformconvergence}
    \max_{t\in I} \Vert \bar\gamma_N(t) - \E\left[\gamma(t)\right]\Vert \rightarrow 0 ~~~ a.s.
    \end{equation}
    Indeed, by Theorem \ref{theorem:PEMisCenterCurve} the PEM of $\gamma$ is unique and in conjunction with the proof of Theorem \ref{theorem:ComputingPEM} and the notation there, there is $\varepsilon_0>0$ such that
    \begin{equation}\label{proof-theorem:PESMuniformconvergence:eq1}
      \big\Vert \E\left[\gamma(t)\right] - \mathcal{F} \big\Vert \geq \varepsilon_0\,,\mbox{for all $t\in I$.  }
    \end{equation}
    Now, if \eqref{equation:uniformconvergence} is true, there is a measurable set $\tilde\Omega$ with $\Prb\big(\tilde\Omega\big)=1$ such that for every $\omega \in \tilde\Omega$ there is  $N_\omega(\varepsilon_0)\in \N$ such that for all $N\geq N_\omega(\varepsilon_0)$ we have that
    \begin{equation*}
      \big\Vert \bar\gamma_N(\omega,t) - \E\left[\gamma(t)\right]\big\Vert < \frac{\varepsilon_0}{2}
    \end{equation*}
    for all $t\in I$. This implies that $\bar\gamma_N(t)\notin\mathcal{F}$ for all $t\in I$, since
    \begin{equation*}
      \big\Vert \bar\gamma_N(t) - \mathcal{F} \big\Vert \geq \big\Vert \E\left[\gamma(t)\right] - \mathcal{F} \big\Vert - \big\Vert \bar\gamma_N(t) - \E\left[\gamma(t)\right]\big\Vert > \frac{\varepsilon_0}{2}
    \end{equation*}
    for all $t\in I$ and all $N>N_\omega(\varepsilon_0)$. Therefore, $\hat E_N(t)$ is a unique point $\hat \mu_N(t)$ for all $N>N_\omega(\varepsilon_0)$. Arguing as in the proof of (ii) of Theorem \ref{theorem:ComputingPEM}, we conclude $\hat \mu_N \in\X$.
    
    Thus, it remains to prove the uniform convergence \eqref{equation:uniformconvergence}. By \citet[Theorem 21.8]{Davidson1994} it suffices to show that the sequence of processes $\bar \gamma_N(t)$ is stochastically equicontinuous, since we already have pointwise convergence by Corollary \ref{corollary:PESMconsistency}.
    
    In order to establish this, consider $\gamma_n(t) = \gamma_0(t)\Exp(\iota\circ A^n_t)$ with  $A_t^n\sim A_t$ for $n=1,\ldots,N$. %with the same distribution as $A$ (see Definition \ref{definition:rGP}). 
    Define $\Omega'$ with $\Prb(\Omega') = 1$ by $\Omega'= \bigcap_{n\in\{1,...,N\}} \Omega_n$, where $\Omega_n$ is the set for which $A^n(\omega)\in \mathcal{C}^1(I,\R^3)$. 
% The triangle inequality then yields
%     \begin{equation}\label{eq:proofUniformConv1}
% 	\Vert \bar\gamma_N(\omega,s) - \bar\gamma_N(\omega,t)\Vert \leq \frac{1}{N}\sum_{n=1}^N \big\Vert \gamma_0(s)\Exp\big(\iota\circ A^n_s(\omega)\big) - \gamma_0(t)\Exp\big(\iota\circ A^n_t(\omega)\big)\big\Vert
%     \end{equation}
%     for all $\omega\in \Omega'$. 
    Using that  $\gamma_0$ and $A^n(\omega)$ are Lipschitz continuous with Lipschitz constants $\max_{t\in I} \Vert \gamma_0'(t) \Vert $ and  $ L_{A^n}(\omega)=\max_{t\in I} \Vert \partial_t A^{n}_t(\omega) \Vert$, respectively, we obtain for all $\omega\in \Omega'$
%     Using that $\Vert R \Vert\leq\sqrt{3}/2$ for all $R\in SO(3)$ and that $\gamma_0$ and $A^n(\omega)$ are Lipschitz continuous with Lipschitz constants $L_\gamma=\max_{t\in I} \Vert \gamma_0'(t) \Vert $ and  $ L_{A^n}(\omega)=\max_{t\in I} \Vert \partial_t A^{n}_t(\omega) \Vert$, respectively, we obtain for all $\omega\in \Omega'$
    \begin{align*}%\nonumber
    \big\Vert \gamma_0(s)\Exp\big(\iota\circ &A^n_s(\omega)\big) - \gamma_0(t)\Exp\big(\iota\circ A^n_t(\omega)\big)\big\Vert \\%\nonumber
    &\hspace*{-1.2cm}\leq \big\Vert\gamma_0(s)- \gamma_0(t)\big\Vert \big\Vert \Exp\big(\iota\circ A^n_s(\omega)\big) \big\Vert 
%     \\\nonumber
%     &~~~~\!
    + \big\Vert\gamma_0(t) \big\Vert \big\Vert \Exp\big(\iota\circ A^n_s(\omega)\big) - \Exp\big(\iota\circ A^n_t(\omega)\big) \big\Vert \\
%     \nonumber
%     &\leq \sqrt{3}/2L_\gamma \vert s - t \vert + \sqrt{3}/2\big\Vert \Exp\big(\iota\circ A^n_s(\omega)\big) - \Exp\big(\iota\circ A^n_t(\omega)\big) \big\Vert\\
    &\hspace*{-1.2cm}\leq M\big( 1  + L_{A^n}(\omega)\big) \vert s - t \vert\,,%\label{eq:proofUniformConv2}
    \end{align*}
    with $M>0$ sufficiently large because  $\Exp$ %:\big(\so(3),\Vert\cdot\Vert\big)\rightarrow \big(SO(3),\Vert\cdot\Vert\big)$ 
    is Lipschitz continuous as well.
%     . Here the last inequality is due to Lipschitz continuity of $\Exp:\big(\so(3),\Vert\cdot\Vert\big)\rightarrow \big(SO(3),\Vert\cdot\Vert\big)$ and Lipschitz continuity of $A^n(\omega)$. 
     In consequence, the triangle inequality then yields
    \begin{align} \nonumber %\label{eq:proofUniformConv1}
	\Vert \bar\gamma_N(\omega,s) - \bar\gamma_N(\omega,t)\Vert &\leq \frac{1}{N}\sum_{n=1}^N \big\Vert \gamma_0(s)\Exp\big(\iota\circ A^n_s(\omega)\big) - \gamma_0(t)\Exp\big(\iota\circ A^n_t(\omega)\big)\big\Vert\\
	\label{eq:LipschitzBoundMean}
      &\leq  M\left(1  + \frac{1}{N}\sum_{n=1}^N L_{A^n}(\omega)\right) \vert s - t \vert
    \end{align}
    for all $\omega\in \Omega'$.

%     Putting \eqref{eq:proofUniformConv1} and \eqref{eq:proofUniformConv2} together yields
%     \begin{equation}\label{eq:LipschitzBoundMean}
%       \Vert \bar\gamma_N(\omega,s) - \bar\gamma_N(\omega,t)\Vert \leq M\left(1  + \frac{1}{N}\sum_{n=1}^N L_{A^n}(\omega)\right) \vert s - t \vert
%     \end{equation}
%       for all $\omega\in \Omega'$. 
      In conjunction with the strong law,
      \begin{equation*}
      \frac{1}{N}\sum_{n=1}^N L_{A^n}(\omega) \rightarrow \E\left[ \max_{t\in I} \Vert \partial_t A_t \Vert \right]\mbox{ as }N\rightarrow\infty
      \end{equation*}
      for all $\omega\in \Omega$ except for a null set, and Assumption \eqref{eq:uniformAssumption}, guaranteeing boundedness of the above r.h.s., (\ref{eq:LipschitzBoundMean}) yields stochastic equicontinuity, finishing the proof of (i).
      
%       We still have to remove the dependency on $\omega$ of the right hand side. Therefore note that by Assumption \eqref{eq:uniformAssumption} and the strong law of large numbers there exists $\Omega''\subset\Omega$ with $\Prb(\Omega'')=1$ such that
%       \begin{equation*}
%       \frac{1}{N}\sum_{n=1}^N L_{A^n}(\omega) \rightarrow \E\left[ \max_{t\in I} \Vert \partial_t A_t \Vert \right]=L<\infty~~~n\rightarrow\infty
%       \end{equation*}
%       for all $\omega\in \Omega''$.
%       Finally, define the new set $\tilde\Omega=\Omega'\cap\Omega''$ still satisfying $\Prb\big(\tilde\Omega\big)=1$. Now, the stochastical equicontinuity follows, since for any $\varepsilon>0$ and any $\omega\in \tilde\Omega$ let $\delta=\tfrac{\varepsilon}{ M(1 + L + \epsilon')} $ with an arbitrary $\epsilon'>0$, we obtain by the SLLN a $N_{\omega,\epsilon'}$ such that
%       \begin{equation*}
%       \left\vert \frac{1}{N}\sum_{n=1}^N L_{A^n}(\omega) - \E\left[ \max_{t\in I} \Vert \partial_t A^{n}_t \Vert \right] \right\vert <\epsilon'
%       \end{equation*}
%       for all $N>N_{\omega,\epsilon'}$. Indeed, combining this with equation \eqref{eq:LipschitzBoundMean} yields
%     \begin{equation*}
%       \sup_{\vert s-t\vert\leq\delta}\Vert \bar\gamma_N(\omega,t) - \bar\gamma_N(\omega,t')\Vert \leq M\left(1  + L + \epsilon' \right) \delta < \varepsilon\,.
%     \end{equation*}
%     for all $N>N_{\omega,\epsilon'}$, which proves the stochastical equicontinuity on $\tilde\Omega$, and therefore finishes the proof of this part.

    (ii): Recalling (\ref{proof-theorem:PESMuniformconvergence:eq1}) define the compact set
    \begin{equation*}
	K = \overline{\bigcup_{t\in I} \Big\{ B\in \R^{3\times3}\,\big\vert~ \Vert \E[ \gamma(t) ] - B \Vert < \tfrac{\varepsilon_0}{2} \Big\}}\,,
    \end{equation*} 
    which satisfies $K\cap \mathcal{F}=\emptyset$. Again by \citet[Theorem  4.1, p. 6]{Dudek1994} we have that $\mathfrak{pr}|_K$ is analytic and hence Lipschitz continuous.
    
%     Since $\E[ \gamma(t) ]$ is unique for all $t\in I$ as deduced in the proof of Theorem \ref{theorem:PEMisCenterCurve}, we define 	analogously to Corollary \ref{corollary:PESMdiffable} the compact set $\Gamma=\big\{ B\in \R^{3\times 3}\,\vert~ \exists t\in I\!: B=\E[ \gamma(t) ] \big\}$. Then $\Gamma \cap \mathcal{F}=\emptyset$, where $\mathcal{F} = \big\{B \in \R^{3\times 3}\,\vert~ {\rm rank}(B)\leq 1\big\}$ and there exists $\varepsilon_0>0$ such that 
%     \begin{equation*}
%       \big\Vert \E\big[\gamma(t)\big] - \mathcal{F} \big\Vert \geq \varepsilon_0\,
%     \end{equation*}
%     for all $t\in I$. Define the compact set
%     \begin{equation*}
% 	K = \bigcup_{t\in I} \Big\{ B\in \R^{3\times3}\,\big\vert~ \Vert \E[ \gamma(t) ] - B \Vert < \tfrac{\varepsilon_0}{2} \Big\}\,,
%     \end{equation*} 
%     which satisfies $K\cap \mathcal{F}=\emptyset$. Again by \citet[Theorem  4.1, p. 6]{Dudek1994} we obtain that the restriction of $\mathfrak{pr}:\big(\R^{3\times 3},\Vert\cdot\Vert\big) \rightarrow (SO(3),\Vert\cdot\Vert)$ to $K$ is analytic and hence Lipschitz continuous.
%     
    Since we proved in (i) that $\bar\gamma_N$ converges a.s. uniformly to $\E[ \gamma ]$, a.s. for $\omega\in \Omega$ there is $N_K(\omega)$ such that $\bar \gamma_N \in K$ for all $N>N_K(\omega)$. Thus, by Lipschitz continuity of $\mathfrak{pr}|_K$, % we have a.s. that
    \begin{equation}\label{eq:wer}
	\big\Vert \hat\mu_N(t) - \gamma_0(t)\big\Vert = \big\Vert \mathfrak{pr}\big(\bar \gamma_N(t)\big) -  \mathfrak{pr}\big(\E[\gamma(t)]\big) \big\Vert \leq C_K \big\Vert \bar \gamma_N(t) - \E[\gamma(t)]\big\Vert
    \end{equation}
    for $C_K>0$ sufficiently large and all $N>N_K(\omega)$, a.s. for $\omega\in \Omega$. Now, the a.s. uniform convergence $\bar\gamma_N\rightarrow \E[ \gamma ]$ implies the a.s. uniform convergence of the left hand side of \eqref{eq:wer} to zero, yielding (ii), in conjunction with (\ref{intrinsic-extrinsic-metrix:eq}).
\qed

\paragraph{Proof of Corollary \ref{corollary:PESMoftenContinuous}.}
%  By the arguments and notations of the proof of Theorem \ref{theorem:PESMuniformconvergence} we have that $t \mapsto \hat\mu\big(\chi_N,t\big)\in \X$ if for any $\varepsilon>0$
%  \begin{equation*}
%   \max_{t\in I} \Vert \bar\gamma_N(t)-\mathcal{F} \Vert > \varepsilon\,. 
%  \end{equation*}
 Using the notation of the proof of Theorem \ref{theorem:PESMuniformconvergence}, with $\varepsilon_0$ there and any $0<\varepsilon<\varepsilon_0$, we have with Theorem \ref{theorem:ComputingPEM} that 
 \begin{equation}
  \begin{aligned} \Prb\big\{ t \mapsto \hat\mu_N(t)\in \X \big\}\geq
    \Prb\left\{ \max_{t\in I} \Vert \bar\gamma_N(t)-\mathcal{F} \Vert > \varepsilon \right\}\\ \geq  \Prb\left\{ \max_{t\in I} \Vert \bar\gamma_N(t)- \E[\gamma(t)] \Vert < \varepsilon_0-\varepsilon \right\}
    & \geq 1 - \frac{\E\big[ \max_{t\in I} \Vert \bar\gamma_N(t)- \E[\gamma(t)]\Vert \big] }{\varepsilon_0-\varepsilon}\,.\label{eq:inequPESM}
  \end{aligned}
 \end{equation}
%  
%  
%  as in the proof of Theorem \ref{theorem:PESMuniformconvergence}. Then for any $\varepsilon_0>\varepsilon>0$ the Markov inequality yields
%  \begin{equation}
%   \begin{aligned}
%     \Prb\left( \max_{t\in I} \Vert \bar\gamma_N(t)-\mathcal{F} \Vert > \varepsilon \right) &\geq  \Prb\left( \max_{t\in I} \Vert \bar\gamma_N(t)- \E[\gamma(t)] \Vert < \varepsilon_0-\varepsilon \right)\\
%     & \geq 1 - \frac{\E\big[ \max_{t\in I} \Vert \bar\gamma_N(t)- \E[\gamma(t)]\Vert \big] }{\varepsilon_0-\varepsilon}\,.\label{eq:inequPESM}
%   \end{aligned}
%  \end{equation}
 In conjunction with (\ref{equation:uniformconvergence}) and $ \Vert \bar\gamma_N(t)- \E[\gamma(t)]  \Vert \leq \sqrt{6}$ for all $t\in I$ (recall that $\|P\| = \sqrt{3/2}$ for all $P\in SO(3)$), Lebesgue's dominated convergence theorem yields \begin{equation*}
   \E\Big[ \max_{t\in I} \Vert \bar\gamma_N(t)- \E[\gamma(t)]\Vert \Big] \rightarrow 0\,,~ ~\text{ for } N\rightarrow \infty\,,
 \end{equation*}
 which together with inequality \eqref{eq:inequPESM} implies the assertion of the corollary.
\qed

\paragraph{Proof of Theorem \ref{theorem:InvariantIntrinsicLoss}.}
   We consider $\delta=\delta_{I,1}$, the proof for the other two losses is similar.
   
   \textit{(i)}:
      Given two curves $\gamma,\eta \in \X$, the length $\delta(\gamma,\eta)$ of the curve $t\mapsto \zeta(t) = \gamma(t)\eta(t)^{-1}$ based on the  intrinsic distance $d_{SO(3)}$ on $SO(3)$ has the representation 
  \begin{equation*}\label{equation:lengthassup}
    \sup\left\{\sum_{k=0}^{K-1} d_{SO(3)}\big(\zeta(t_k),\zeta(t_{k+1})\big) ~ \big\vert ~ K \in \mathbb{N},~ 0=t_0<t_1<...<t_{k-1}<t_K=1 \right\}\,,
  \end{equation*}
  where (\ref{intrinsic-extrinsic-metrix:eq}) implies $d_{SO(3)}\big(\zeta(t_k),\zeta(t_{k+1})\big) = d_{SO(3)}\big(\zeta(t_{k})^{-1},\zeta(t_{k+1})^{-1}\big)$ yielding symmetry of $\delta$.~\\
   \textit{(ii)} follows directly from biinvariance and from the fact that the length of a curve does not depend on its parametrization. % (e.g., \citet[Theorem 13.25, p.338]{Lee2013}).
   
   \textit{(iii)}: $\delta(P\gamma,\eta)=\delta(\gamma,\eta)$ is again a direct consequence of the biinvariance of the metric.
    
    \textit{(iv)}: Assume $\eta = P\gamma$. Then by (iii) we have %and the first statement of this theorem we obtain
    \begin{equation*}
     \delta(\gamma,\eta) = \delta(\gamma,P\gamma) = \delta(\gamma,\gamma) = {\rm length}(\gamma\gamma^{-1}) =0\,.
    \end{equation*}
    For the other way round assume that $\delta(\gamma,\eta)=0$. In consequence, $\gamma\eta^{-1}$ is a constant, $P^{-1}\in SO(3)$, say, yielding $P \gamma =\eta$.
%     n$\eta\neq P\gamma$ for all $P\in SO(3)$. Defining $\zeta=\gamma\eta^{-1}$  $\xi = \zeta\zeta(0)^{-1}$, by assumption, 
%     
%     and setting 
%     Then define the continuously differentiable curve $\zeta=\gamma\eta^{-1}$ and $\zeta(0)=\zeta_0$. The curve $\zeta\zeta_0^{-1}$ is continuously differentiable as well and by assumption $\zeta_0^{-1}\zeta$ cannot be constantly $I_{3\times 3}$ for all $t\in I$. Therefore we find an $t'\in I$ such that $\zeta_0^{-1}\zeta(t')=R\neq I_{3\times3}$ for some $R\in SO(3)$. Thus, we have that
%     \begin{equation*}
%       0 < d_{SO(3)}\big(I_{3\times 3}, R\big) \leq {\rm length}\big(\zeta_0^{-1}\zeta\vert_{[0,t']}\big) = {\rm length}\big(\zeta\vert_{[0,t']}\big) \leq {\rm length}\big(\zeta\big)\,,
%     \end{equation*}
%     where the equality stems from the bi-invariance of the metric.
%     
%     
% \qed
% 
% \paragraph{Proof of Theorem \ref{theorem:InverseAlignmentDiff}.}
    The proof of \textit{(v)} is straightforward.
%     \begin{align*}
%       \inf_{\phi\in \Diff} \delta\!\left( \gamma_1, \gamma_2\circ\phi \right) &= \inf_{g\in \Diff} \delta\!\left(  \gamma_1\circ\phi^{-1},\gamma_2 \right) \\
%       &= \inf_{g\in \Diff} \delta\!\left( \gamma_2 ,\gamma_1\circ\phi^{-1} \right) \\
% 						&= \inf_{\phi'=\phi^{-1}\in \Diff} \delta\!\left( \gamma_2 ,\gamma_1\circ\phi' \right)\,,
%     \end{align*}
%     where we used in the first equality the symmetry and $\Diff$-invariance of $\delta$ and in the second equality the properties of a group action.
\qed

\paragraph{Proof of Theorem \ref{theorem:Lift}.}
% \textbf{Proof of Theorem \ref{theorem:Lift}:}
    Let $\gamma\in \X$ and $(P,Q)\in SO(3)\times SO(3)$ be arbitrary and choose continuous lifts $\tilde\gamma$ and $\widetilde{P\gamma Q}$ of $\gamma$ and $P\gamma Q$. Since $\pi$ is a homomorphism, we have that
    \begin{align*} \pi(p\tilde \gamma {q})= \pi\left(\widetilde{P\gamma Q}\right) 
    \end{align*}
    with any $p \in \pi^{-1}(P)$ and $q \in \pi^{-1}(Q)$. Since $\pi$ is a double cover, there is a function $\epsilon: [0,1]\rightarrow\{-1,1\}$ such that
    \begin{align*} 
    p\tilde \gamma {q} = \epsilon \widetilde{P\gamma Q}\,.
    \end{align*}
    By continuity of the lifts, $\epsilon$ must be constant and hence by virtue of (iii) of Lemma \ref{lemma:Quat2Rot} there is a unique $R\in SO(4)$ such that
%     \begin{align*} 
  $  R\tilde \gamma = \epsilon p\tilde \gamma {q}^{} = \widetilde{P\gamma Q}$
%     \end{align*}
    yielding assertion (i),(a). Then $(\pi,\pi) \circ r_{SO(4)} R = (\pi,\pi)\big(\kappa( p,q)\big) =(P,Q)$ with $\kappa \in \{-1,1\}$ yields assertion (i),(b).
    
    Assertion (ii) follows at once from the following
        \begin{align*} 
     \min_{P,Q\in SO(3)} L(P\gamma Q,\eta) \geq \min_{R\in SO(4)}\wL\left(R\tilde\gamma,\tilde \eta\right) \geq  \min_{P,Q\in SO(3) }\wL\left(\widetilde{P\gamma Q},\tilde \eta\right)\geq  \min_{P,Q\in SO(3)} L(P\gamma Q,\eta)\,.
     \end{align*}
     Here the first inequality is due to the definition (taking into account that $-I_{4\times 4}\in SO(4)$),  
    \begin{align*} 
     L(P\gamma Q,\eta) = \min\left\{\wL\left(\widetilde{P\gamma Q},\tilde \eta\right), \wL\left(-\widetilde{P\gamma Q},\tilde \eta\right)\right\}\,,
    \end{align*}
    and (i),(a); the second follows from the consideration that 
    for any $R\in SO(4)$ we have $\pi(R\tilde \gamma) = P \gamma Q$ with $(P,Q)=\Pi(R)$, such that a continuous lift can be chosen such that $\widetilde{P\gamma Q} =R\tilde \gamma$ ;   and the third holds again due to the definition of $\wL$.
    
%     Here is the consideration: For any $R\in SO(4)$ we have $\pi(R\tilde \gamma) = P \gamma Q$ with $(P,Q)=\Pi(R)$. Hence, a continous lift can be chosen such that $\widetilde{P\gamma Q} =R\tilde \gamma$. 
     
     Assertion (iii) is straightforward.
\qed

\paragraph{Proof of Theorem \ref{theorem:GeneralConsistency}.}
% and Corollary \ref{corollary:GeneralConsistency}.}
Consider the function
\begin{align*} f:  \mathbb{X} \times \mathbb{Y} \to \RR,~ \big((P,Q),(\gamma_1,\gamma_2)\big)\mapsto L(P\gamma_1Q,\gamma_2)
\end{align*}
mapping from the compact space $\mathbb{X} = \Inull$ and $\mathbb{Y} = \X\times\X$. By hypothesis and Theorem \ref{theorem:computingMPestim}, for fixed $y_0=(\gamma_0,\eta_0)$ it has a unique minimizer $x^*=(P^*,Q^*) \in \mathbb X$. Moreover, by Theorem \ref{theorem:PESMuniformconvergence}, for any sample PEM curves, $y=(\hat\gamma_N,\hat\eta_M)$ converges to $y_0$ outside a set of measure zero and both sample PEM curves belong a.s. eventually to $\X$ by Corollary \ref{corollary:PESMoftenContinuous}. By Lemma \ref{Lem:Max} below (in \citet[Lemma 2]{Chang1986}), every minimizer $x=(\hat P_N,\hat Q_M)$ of $f\big((P,Q),(\hat\gamma_N,\hat\eta_M)\big)$ converges to $x^*=(P^*,Q^*)$, again outside a set of measure zero.\qed

% The next Lemma is a useful tool showing strong consistency, if one has $M$-estimators, and was stated in the prove of 
% {\bf \citet[Lemma 2]{Chang1986} has a proof but seems quite remote from our problem. I would drop the reference}
\begin{Lem}\label{Lem:Max}
 Let $f: \mathbb{X} \times \mathbb{Y} \rightarrow \mathbb{R}$ be continuous with $\mathbb{X}$ compact, and assume for fixed $y_0\in \mathbb{Y}$ that $x\mapsto f(x,y_0)$ has a unique minimizer $x^*$. Moreover suppose that $y_n$ converges to $y_0$ and that each $x_n$ is a minimizer of $x\mapsto f(x,y_n)$. Then $x_n \rightarrow x^*$.
\end{Lem}
\begin{proof}
 By compactness, $x_n$ has cluster points. Since by hypothesis $f(x_{n},y_{n})\leq f(x,y_{n})$ for all $n\in \mathbb N$ and $x\in\mathbb X$, by continuity and uniqueness, every cluster point is $x^*$.
%  Let $\left(x_{n_k}\right)$ be any convergent subsequence of $x_n$. By the definition of $x_n$ we have that $f(x_{n_k},y_{n_k})\geq f(x,y_{n_k})$ for all $k$ and all $x \in X$. Hence if $k$ tends to infinity by the continuity of $f$ and the uniqueness of $x_0$ we have that $x_{n_k}$ converges to $x_0$.
%  
%  Now assume $x_n$ does not converge to $x_0$. Then there is a subsequence $(x_{n_j})$ such that $\vert x_{n_j}-x_0\vert > \varepsilon$. But this is a contradiction, since by compactness this sequence again has a convergent subsequence, which by the above argumentation, converges to $x_0$. Thus $x_n$ must converge to $x_0$.
\end{proof}

\bibliography{literature}{}

\begin{thebibliography}{}

\bibitem[\protect\citeauthoryear{Amor, Su, and Srivastava}{Amor
  et~al.}{2016}]{Amor2016}
Amor, B.~B., J.~Su, and A.~Srivastava (2016).
\newblock Action recognition using rate-invariant analysis of skeletal shape
  trajectories.
\newblock {\em Pattern Analysis and Machine Intelligence, IEEE Transactions
  on\/}~{\em 38\/}(1), 1--13.

\bibitem[\protect\citeauthoryear{Bauer, Bruveris, and Michor}{Bauer
  et~al.}{2016}]{Bauer2016}
Bauer, M., M.~Bruveris, and P.~W. Michor (2016).
\newblock Why use {S}obolev metrics on the space of curves?
\newblock In {\em Riemannian Computing in Computer Vision}, pp.\  233--255.
  Springer.

\bibitem[\protect\citeauthoryear{Bauer, Eslitzbichler, and Grasmair}{Bauer
  et~al.}{2015}]{Bauer2015}
Bauer, M., M.~Eslitzbichler, and M.~Grasmair (2015).
\newblock Landmark-guided elastic shape analysis of human character motions.
\newblock {\em arXiv preprint arXiv:1502.07666\/}.

\bibitem[\protect\citeauthoryear{Bhattacharya and Patrangenaru}{Bhattacharya
  and Patrangenaru}{2003}]{Bhattacharya2003}
Bhattacharya, R.~N. and V.~Patrangenaru (2003).
\newblock Large sample theory of intrinsic and extrinsic sample means on
  manifolds {I}.
\newblock {\em The Annals of Statistics\/}~{\em 31\/}(1), 1--29.

\bibitem[\protect\citeauthoryear{Celledoni and Eslitzbichler}{Celledoni and
  Eslitzbichler}{2015}]{Celledoni2015}
Celledoni, E. and M.~Eslitzbichler (2015).
\newblock Shape analysis on lie groups with applications in computer animation.
\newblock {\em arXiv preprint arXiv:1506.00783\/}.

\bibitem[\protect\citeauthoryear{Chang}{Chang}{1986}]{Chang1986}
Chang, T. (1986).
\newblock Spherical regression.
\newblock {\em The Annals of Statistics\/}~{\em 14\/}(3), 907--924.

\bibitem[\protect\citeauthoryear{Chirikjian and Kyatkin}{Chirikjian and
  Kyatkin}{2000}]{Chirikjian2000}
Chirikjian, G.~S. and A.~B. Kyatkin (2000).
\newblock {\em Engineering applications of noncommutative harmonic analysis:
  with emphasis on rotation and motion groups}.
\newblock CRC press.

\bibitem[\protect\citeauthoryear{Davidson}{Davidson}{1994}]{Davidson1994}
Davidson, J. (1994).
\newblock {\em Stochastic limit theory: An introduction for econometricians}.
\newblock OUP Oxford.

\bibitem[\protect\citeauthoryear{Downs}{Downs}{1972}]{Downs1972}
Downs, T.~D. (1972).
\newblock Orientation statistics.
\newblock {\em Biometrika\/}~{\em 59\/}(3), 665--676.

\bibitem[\protect\citeauthoryear{Dudek and Holly}{Dudek and
  Holly}{1994}]{Dudek1994}
Dudek, E. and K.~Holly (1994).
\newblock Nonlinear orthogonal projection.
\newblock In {\em Annales Polonici Mathematici}, Volume~59, pp.\  1--31.
  Institute of Mathematics Polish Academy of Sciences.

\bibitem[\protect\citeauthoryear{Duhamel, Bourriez, Devos, Krystkowiak, Destee,
  Derambure, and Defebvre}{Duhamel et~al.}{2004}]{duhamel2004statistical}
Duhamel, A., J.~Bourriez, P.~Devos, P.~Krystkowiak, A.~Destee, P.~Derambure,
  and L.~Defebvre (2004).
\newblock Statistical tools for clinical gait analysis.
\newblock {\em Gait \& posture\/}~{\em 20\/}(2), 204--212.

\bibitem[\protect\citeauthoryear{Fletcher}{Fletcher}{2011}]{Fletcher2011}
Fletcher, T. (2011).
\newblock Geodesic regression on riemannian manifolds.
\newblock In {\em Proceedings of the Third International Workshop on
  Mathematical Foundations of Computational Anatomy-Geometrical and Statistical
  Methods for Modelling Biological Shape Variability}, pp.\  75--86.

\bibitem[\protect\citeauthoryear{Groen, Geurts, Nienhuis, and Duysens}{Groen
  et~al.}{2012}]{groen2012sensitivity}
Groen, B., M.~Geurts, B.~Nienhuis, and J.~Duysens (2012).
\newblock Sensitivity of the olga and vcm models to erroneous marker placement:
  effects on 3d-gait kinematics.
\newblock {\em Gait \& posture\/}~{\em 35\/}(3), 517--521.

\bibitem[\protect\citeauthoryear{Grood and Suntay}{Grood and
  Suntay}{1983}]{grood1983joint}
Grood, E.~S. and W.~J. Suntay (1983).
\newblock A joint coordinate system for the clinical description of
  three-dimensional motions: application to the knee.
\newblock {\em Journal of biomechanical engineering\/}~{\em 105\/}(2),
  136--144.

\bibitem[\protect\citeauthoryear{Helgason}{Helgason}{1962}]{Helgason1962}
Helgason, S. (1962).
\newblock {\em Differential Geometry and Symmetric Spaces}.
\newblock New York: Academic Press.

\bibitem[\protect\citeauthoryear{Huckemann}{Huckemann}{2011}]{Huckemann2011b}
Huckemann, S. (2011).
\newblock Inference on 3d {P}rocrustes means: Tree bole growth, rank deficient
  diffusion tensors and perturbation models.
\newblock {\em Scandinavian Journal of Statistics\/}~{\em 38\/}(3), 424--446.

\bibitem[\protect\citeauthoryear{Huckemann}{Huckemann}{2012}]{Huckemann12}
Huckemann, S. (2012).
\newblock On the meaning of mean shape: Manifold stability, locus and the two
  sample test.
\newblock {\em Annals of the Institute of Statistical Mathematics\/}~{\em
  64\/}(6), 1227--1259.

\bibitem[\protect\citeauthoryear{Kent and Mardia}{Kent and
  Mardia}{1997}]{Kent1997}
Kent, J.~T. and K.~V. Mardia (1997).
\newblock Consistency of {P}rocrustes estimators.
\newblock {\em Journal of the Royal Statistical Society: Series B (Statistical
  Methodology)\/}~{\em 59\/}(1), 281--290.

\bibitem[\protect\citeauthoryear{Kurtek, Srivastava, and Wu}{Kurtek
  et~al.}{2011}]{Kurtek2011}
Kurtek, S.~A., A.~Srivastava, and W.~Wu (2011).
\newblock Signal estimation under random time-warpings and nonlinear signal
  alignment.
\newblock In J.~Shawe-Taylor, R.~Zemel, P.~Bartlett, F.~Pereira, and
  K.~Weinberger (Eds.), {\em Advances in Neural Information Processing Systems
  24}, pp.\  675--683. Curran Associates, Inc.

\bibitem[\protect\citeauthoryear{Leardini, Chiari, Della~Croce, and
  Cappozzo}{Leardini et~al.}{2005}]{Leardini2005}
Leardini, A., L.~Chiari, U.~Della~Croce, and A.~Cappozzo (2005).
\newblock Human movement analysis using stereophotogrammetry: {p}art 3. soft
  tissue artifact assessment and compensation.
\newblock {\em Gait \& Posture\/}~{\em 21\/}(2), 212--225.

\bibitem[\protect\citeauthoryear{Lee}{Lee}{2013}]{Lee2013}
Lee, J.~M. (2013).
\newblock {\em Introduction to Smooth manifolds}.
\newblock Springer Verlag, New York.

\bibitem[\protect\citeauthoryear{Mebius}{Mebius}{2005}]{Mebius2005}
Mebius, J.~E. (2005).
\newblock A matrix-based proof of the quaternion representation theorem for
  four-dimensional rotations.
\newblock {\em arXiv preprint math/0501249\/}.

\bibitem[\protect\citeauthoryear{Myers and Steenrod}{Myers and
  Steenrod}{1939}]{MyersSteenrood1939}
Myers, S.~B. and N.~Steenrod (1939).
\newblock The group of isometries of a {R}iemannian manifold.
\newblock {\em Annals of Mathematics\/}~{\em 40\/}(2), 400--416.

\bibitem[\protect\citeauthoryear{Noehren, Manal, and Davis}{Noehren
  et~al.}{2010}]{noehren2010improving}
Noehren, B., K.~Manal, and I.~Davis (2010).
\newblock Improving between-day kinematic reliability using a marker placement
  device.
\newblock {\em Journal of Orthopaedic Research\/}~{\em 28\/}(11), 1405--1410.

\bibitem[\protect\citeauthoryear{{\~O}unpuu, Solomito, Bell, DeLuca, and
  Pierz}{{\~O}unpuu et~al.}{2015}]{ounpuu2015long}
{\~O}unpuu, S., M.~Solomito, K.~Bell, P.~DeLuca, and K.~Pierz (2015).
\newblock Long-term outcomes after multilevel surgery including rectus femoris,
  hamstring and gastrocnemius procedures in children with cerebral palsy.
\newblock {\em Gait \& posture\/}~{\em 42\/}(3), 365--372.

\bibitem[\protect\citeauthoryear{Pierrynowski, Costigan, Maly, and
  Kim}{Pierrynowski et~al.}{2010}]{pierrynowski2010patients}
Pierrynowski, M.~R., P.~A. Costigan, M.~R. Maly, and P.~T. Kim (2010).
\newblock Patients with osteoarthritic knees have shorter orientation and
  tangent indicatrices during gait.
\newblock {\em Clinical Biomechanics\/}~{\em 25\/}(3), 237--241.

\bibitem[\protect\citeauthoryear{Ramsay}{Ramsay}{1982}]{Ramsay1982}
Ramsay, J. (1982).
\newblock When the data are functions.
\newblock {\em Psychometrika\/}~{\em 47\/}(4), 379--396.

\bibitem[\protect\citeauthoryear{Rancourt, Rivest, and Asselin}{Rancourt
  et~al.}{2000}]{Rancourt2000}
Rancourt, D., L.-P. Rivest, and J.~Asselin (2000).
\newblock Using orientation statistics to investigate variations in human
  kinematics.
\newblock {\em Journal of the Royal Statistical Society: Series C (Applied
  Statistics)\/}~{\em 49\/}(1), 81--94.

\bibitem[\protect\citeauthoryear{Srivastava, Klassen, Joshi, and
  Jermyn}{Srivastava et~al.}{2011}]{SrivastavaKlassen2011}
Srivastava, A., E.~Klassen, S.~H. Joshi, and I.~H. Jermyn (2011).
\newblock Shape analysis of elastic curves in euclidean spaces.
\newblock {\em Pattern Analysis and Machine Intelligence, IEEE Transactions
  on\/}~{\em 33\/}(7), 1415--1428.

\bibitem[\protect\citeauthoryear{Srivastava, Wu, Kurtek, Klassen, and
  Marron}{Srivastava et~al.}{}]{SrivastavaWu2011}
Srivastava, A., W.~Wu, S.~Kurtek, E.~Klassen, and J.~S. Marron.
\newblock Registration of functional data using fisher-rao metric.
\newblock {\em arXiv preprint arXiv:1103.3817\/}.

\bibitem[\protect\citeauthoryear{Stanfill, Genschel, and Hofmann}{Stanfill
  et~al.}{2013}]{Stanfill2013}
Stanfill, B., U.~Genschel, and H.~Hofmann (2013).
\newblock Point estimation of the central orientation of random rotations.
\newblock {\em Technometrics\/}~{\em 55\/}(4), 524--535.

\bibitem[\protect\citeauthoryear{Stuelpnagel}{Stuelpnagel}{1964}]{Stuelpnagel1964}
Stuelpnagel, J. (1964).
\newblock On the parametrization of the three-dimensional rotation group.
\newblock {\em SIAM review\/}~{\em 6\/}(4), 422--430.

\bibitem[\protect\citeauthoryear{Su, Kurtek, Klassen, Srivastava, et~al.}{Su
  et~al.}{2014}]{SuKurtek2014}
Su, J., S.~Kurtek, E.~Klassen, A.~Srivastava, et~al. (2014).
\newblock Statistical analysis of trajectories on {R}iemannian manifolds:
  {b}ird migration, hurricane tracking and video surveillance.
\newblock {\em The Annals of Applied Statistics\/}~{\em 8\/}(1), 530--552.

\bibitem[\protect\citeauthoryear{Umeyama}{Umeyama}{1991}]{Umeyama1991}
Umeyama, S. (1991).
\newblock Least-squares estimation of transformation parameters between two
  point patterns.
\newblock {\em Pattern Analysis and Machine Intelligence, IEEE Transactions
  on\/}~{\em 13\/}(4), 376--380.

\bibitem[\protect\citeauthoryear{Wang, Chiou, and M{\"u}ller}{Wang
  et~al.}{2015}]{Wang2015}
Wang, J.-L., J.-M. Chiou, and H.-G. M{\"u}ller (2015).
\newblock Review of functional data analysis.
\newblock {\em Annu. Rev. Statist\/}~{\em 1}, 41.

\bibitem[\protect\citeauthoryear{Ziezold}{Ziezold}{1977}]{Ziezold77}
Ziezold, H. (1977).
\newblock Expected figures and a strong law of large numbers for random
  elements in quasi-metric spaces.
\newblock {\em Transaction of the 7th Prague Conference on Information Theory,
  Statistical Decision Function and Random Processes\/}~{\em A}, 591--602.

\end{thebibliography}
\bibliographystyle{Chicago}

% \bibliographystyle{../../BIB/Chicago} 
% \bibliography{../../BIB/shape,../../BIB/stochgeom,../../BIB/diffgeo,../../BIB/stats,../../BIB/biomech,../../BIB/numerics}
\end{document}